\documentclass[conference]{IEEEtran}
\IEEEoverridecommandlockouts

\usepackage{float}
    
\usepackage[cmex10]{mathtools}
\usepackage{mathrsfs}
\usepackage[export]{adjustbox}

\usepackage[noadjust]{cite}
\usepackage{amsmath,amssymb,amsfonts}

\newcommand*{\Scale}[2][4]{\scalebox{#1}{$#2$}}%

\usepackage{mathtools}
\usepackage{algorithmic}
\usepackage{graphicx}
\usepackage{textcomp}
\usepackage{amsthm}
\usepackage{graphicx}
\usepackage{bm}
\usepackage{dsfont}
\usepackage{enumitem}
\usepackage{comment}

\usepackage{subcaption}

\usepackage{color,soul}
\usepackage{multirow,array}
\usepackage{blkarray}
\graphicspath{{./images/}}
\usepackage{amssymb}
\usepackage{amsmath}

\newtheorem{theorem}{Theorem}
\newtheorem{lemma}{Lemma}
\newtheorem{remark}{Remark}

\newtheorem{corollaryLemma}{Corollary}[lemma]
\newtheorem{assumption}{Assumption}
\newtheorem{definition}{Definition}

\usepackage[compact]{titlesec}
\titlespacing{\section}{0pt}{2ex}{1ex}
\titlespacing{\subsection}{0pt}{1ex}{0ex}
\titlespacing{\subsubsection}{0pt}{0.5ex}{0ex}

\newtheorem{innercustomgeneric}{\customgenericname}
\providecommand{\customgenericname}{}
\newcommand{\newcustomtheorem}[2]{%
  \newenvironment{#1}[1]
  {%
   \renewcommand\customgenericname{#2}%
   \renewcommand\theinnercustomgeneric{##1}%
   \innercustomgeneric
  }
  {\endinnercustomgeneric}
}
\newcustomtheorem{customclaim}{Claim}

\usepackage[T1]{fontenc}
\usepackage{tabularx,ragged2e,booktabs}

\newcommand*{\set}{\fontfamily{qag}\selectfont}
\DeclareTextFontCommand{\textset}{\set}

\usepackage{ifthen}
\newboolean{extend_v}
\newboolean{editor}

\begin{document}

\setboolean{extend_v}{true} 

\title{Derandomizing Codes for the Binary Adversarial Wiretap Channel of Type II \\
{\thanks{This work is supported in part by the U.S. National Science Foundation under grants CNS-2128448, CNS-2212565, CNS-2225577, EEC-1941529, ITE-2226447 and by the Office of Naval Research under grant ONR N000142112472.}}}

\author{\IEEEauthorblockN{Eric Ruzomberka\IEEEauthorrefmark{1}, Homa Nikbakht\IEEEauthorrefmark{1}, Christopher G. Brinton\IEEEauthorrefmark{2}, David J. Love\IEEEauthorrefmark{2} and H. Vincent Poor\IEEEauthorrefmark{1}}
\IEEEauthorblockA{ \IEEEauthorrefmark{1}\textit{Princeton University \IEEEauthorrefmark{2}Purdue University }}
}

\maketitle

\thispagestyle{plain}
\pagestyle{plain}

\begin{abstract}
We revisit the binary adversarial wiretap channel (AWTC) of type II in which an active adversary can read a fraction $r$ and flip a fraction $p$ of codeword bits. The semantic-secrecy capacity of the AWTC II is partially known, where the best-known lower bound is non-constructive, proven via a random coding argument that uses a large number (that is exponential in blocklength $n$) of random bits to seed the random code. In this paper, we establish a new derandomization result in which we match the best-known lower bound of $1-H_2(p)-r$ where $H_2(\cdot)$ is the binary entropy function via a random code that uses a small seed of only $O(n^2)$ bits. Our random code construction is a novel application of \textit{pseudolinear codes} -- a class of non-linear codes that have $k$-wise independent codewords when picked at random where $k$ is a design parameter. As the key technical tool in our analysis, we provide a soft-covering lemma in the flavor of Goldfeld, Cuff and Permuter (Trans. Inf. Theory 2016) that holds for random codes with $k$-wise independent codewords. 
\end{abstract}

 \section{Introduction}

 Consider a communication setting in which a sender Alice wishes to communicate a message to a receiver Bob by sending a sequence of bits over a noisy wiretap channel. The channel is controlled by an (active) adversary who can both read a fraction $r \in [0,1]$ and flip a fraction $p \in [0,1]$ of Alice’s transmitted bits. In this setting, Alice’s and Bob’s communication goal under any adversarial strategy is two-fold:
 \begin{enumerate}
\item	(\textit{Reliability}) Bob must decode Alice’s message with small probability of error.
\item	(\textit{Secrecy}) The adversary must extract negligible information of the Alice's message via its observation of Alice's sequence.
\end{enumerate}
Critically, we make no assumptions about the adversary’s computational limitations, and thus, secrecy must be guaranteed in an information theoretic sense by ``hiding'' the message in the adversary's bit-limited observation. Furthermore, the adversary may choose the location of the bit reads and bit flips in an arbitrary manner using knowledge of Alice and Bob communication scheme. In the literature, the above setting is known as the binary adversarial wiretap channel of type II (denoted as $(p,r)$-AWTC II) \cite{Ozarow1986,Wang2016_2}.

Much is known about the fundamental limits of communication over the $(p,r)$-AWTC II. Roughly defined, the secrecy capacity of the $(p,r)$-AWTC II is the largest rate at which Alice and Bob can communicate while meeting the above goals under a given secrecy measure. The measure we focus on is semantic secrecy (SS) \cite{Goldwasser1984,Bellare1997}, which is widely recognized as the cryptographic gold standard for evaluating secrecy \cite{Bellare2012}. The SS capacity, denoted $C(p,r)$, is partially known where the best-known lower bound \cite{Wang2016} and upper bound \cite{Wang2016,Ozel2013} are
\begin{equation} \label{eq:SSC_ulb}
\max \{1-H_2(p) - r,0 \} \leq C(p,r) \leq 1-H_2(p) - r - \min\limits_{x \in [0,1]} f(x)
\end{equation}
where $H_2(\cdot)$ is the binary entropy function and $f(x) = H_2((2p-1)x+1-p) - H_2(p) - rH_2(x)$. Note that the two bounds are close for small $r$ and tight for $p=0$.

While the limits of communication over the $(p,r)$-AWTC II are mostly understood, less is known on how to construct efficient codes to achieve these limits. The proof of the lower bound (\ref{eq:SSC_ulb}), as given in \cite{Wang2016}, is non-constructive and follows an ordinary random coding argument in which codewords are chosen uniformly and \textit{independently} from space $\{0,1 \}^n$ where $n$ is the blocklength of the code. As a tool for probabilistic constructions, the practical use of this random code distribution is limited. For example, to represent a code picked in this way, one would need to remember at least $n 2^{Rn}$ random bits where $R$ is the coding rate.\footnote{Additional random (seed) bits are needed if one considers codes with private randomness at the encoder (i.e., stochastic codes).} Thus, codes picked from a distribution with mutual independence property lack a succinct representation. Furthermore, the high degree of randomness used in the construction obscures insight into the structure of a good code. Without sufficient structure, efficient encoding and decoding algorithms are likely to be elusive.   

In this paper, we work towards an efficient code construction for the $(p,r)$-AWTC II by partially derandomizing the random code used in \cite{Wang2016} to establish the lower bound (\ref{eq:SSC_ulb}). We do so by relaxing the requirement that codewords be mutually independent and consider random codes with $k$-wise independent codewords for some positive integer $k << n$. We show that random codes under this weaker notation of independence can achieve the lower bound (\ref{eq:SSC_ulb}) for some parameter $k$ large enough but constant in $n$. As a result, these codes have both a more succinct representation and additional structure compared to random codes with mutually independent codewords. 

The approach we take is the following. We focus on a class of non-linear codes known as \textit{pseudolinear codes} (precisely defined in Section \ref{sec:results}), which was initially proposed by Guruswami and Indyk \cite{Guruswami2001_2} outside of the AWTC setting. In the AWTC setting, pseudolinear codes have a number of nice properties, including succinct representations (i.e., $O(k n^2)$ bits), efficient encoding algorithms, some linear structure, and $k$-wise independent codewords when chosen at random for a designable parameter $k$. We initiate the study of pseudolinear codes for achieving both secrecy and reliability in the wiretap setting. As our main result, we show that random pseudolinear codes achieve the best-known SS capacity lower bound (\ref{eq:SSC_ulb}). Conversely, we show that non-linear codes are \textit{necessary} to achieve this lower bound for some values of $p$ and $r$. To prove our main result, we provide a new lemma on the soft-covering phenomenon \cite{Wyner1975,Goldfeld2016} under random coding with $k$-wise independent codewords.

 \section{Preliminaries, Results \& Related Work}

 \subsection{Notation}

 Unless stated otherwise, we denote random variables in uppercase font (e.g., $X$), realizations of random variables in lowercase font (e.g., $x$), and sequences in bold font (e.g., $\bm{X}$, $\bm{x}$). An exception to the above rules occurs when we denote codes: we denote random codes with script typeset (e.g., $\mathscr{C}$) and realizations of random codes with calligraphic typeset (e.g., $\mathcal{C}$). We denote the set of all possible distributions over a set $\mathcal{X}$ as $\mathcal{P}(\mathcal{X})$, and denote the uniform distribution over $\mathcal{X}$ as $\mathrm{Unif}(\mathcal{X})$. We denote that $X$ is distributed as $P \in \mathcal{P}(\mathcal{X})$ by writing $X \sim P$. For PMFs $P$ and $Q$ such that $\mathrm{supp}(P) \subseteq \mathrm{supp}(Q)$ (absolute continuity), the relative entropy of $P$ and $Q$ is
 $D(P||Q) \triangleq \sum_{x \in \mathrm{supp}(P)} P(x) \log_2 \frac{P(x)}{Q(x)}$. For $\alpha>0$ and $\alpha \neq 1$, the R\'{e}nyi divergence of order $\alpha$ is
 $D_{\alpha}(P||Q) \triangleq \frac{1}{\alpha-1} \log_2 \sum_{x \in \mathrm{supp}(P)} P(x) (\frac{P(x)}{Q(x)})^{\alpha-1}$. Define the special case $D_1(P||Q) \triangleq \lim_{\alpha \rightarrow 1} D_\alpha(P||Q) = D(P||Q)$. For an event $\mathcal{A}$, we let $\mathds{1} \{ \mathcal{A} \}$ denote the indicator of $\mathcal{A}$.

 \subsection{Setup} \label{sec:setup}

 \textit{Code:} A (binary) code $\mathcal{C}_n$ of blocklength $n$ is a subset of $\{0,1\}^n$. We will associate a code $\mathcal{C}_n$ with an encoding function $\bm{x}(\cdot)$, which performs a mapping from the message space $\mathcal{M}$ to codewords in $\{0,1\}^n$. As is common for wiretap codes, we will consider \textit{stochastic encoding} in which $\bm{x}$ takes as argument a private random key $w \in \mathcal{W}$ that is known only to Alice. Specifically, for a message rate $R = \frac{\log_2 |\mathcal{M}|}{n}$ and a (private) key rate $R' = \frac{\log_2 |\mathcal{W}|}{n}$, an $[n,R n,R' n]$ code $\mathcal{C}_n$ is a set $$\mathcal{C}_n = \{ \bm{x}(m,w): (m,w) \in \mathcal{M} \times \mathcal{W}\}$$ where we refer to $\bm{x}(w,m)$ as the ($n$-bit) codeword corresponding to message $m$ and key $w$. In turn, a \textit{family} of codes is a sequence $\{\mathcal{C}_n\}_{n=1}^{\infty}$ where for each $n\geq1$, $\mathcal{C}_n$ is an $[n,Rn,R'n]$ code.

 \textit{Encoding/Decoding:} For an $[n,Rn,R'n]$ code $\mathcal{C}_n$, probability mass function (PMF) $P_M \in \mathcal{P}(\mathcal{M})$, a message $M \sim P_M$ and a private key $W \sim \mathrm{Unif}(\mathcal{W})$ where $M$ and $W$ are independent, Alice encodes $M$ into a codeword $\bm{x}(M,W)$ and transmits it over the channel. Subsequently, Bob receives a corrupted version of the codeword and performs decoding by choosing a message estimate $\widehat{M} \in \mathcal{M}$. We say that a decoding error occurs if $\widehat{M} \neq M$.

 \textit{The AWTC II:}
 For a read fraction $r \in [0,1]$ and an error fraction $p \in [0,1/2]$, the adversary can observe $rn$ bits and flip up to $pn$ bits of $\bm{x}(M,W)$. The location of the read bits are indexed by a coordinate set $\mathcal{S}$, which the adversary can choose from the set $\mathscr{S}$ consisting of all subsets of $[n]$ of size $rn$. In turn, the adversary observes $\bm{Z} = \bm{x}(M,W,\mathcal{S})$ where $\bm{x}(M,W,\mathcal{S})$ denotes the $rn$ bits of $\bm{x}(M,W)$ indexed at $\mathcal{S}$, and subsequently, chooses the location of the bit flips. We emphasize that the location of the bit flips need not coincide with $\mathcal{S}$. In general, the adversary can randomize its above choices by choosing a distribution on $\mathcal{S}$ that can depend on the code, as well as a distribution on the bit flip locations that can depend on both the code and the observation $\bm{Z}$.
 
 \textit{Secrecy:} Define the semantic leakage as
 \begin{equation}
 \mathrm{Sem}(\mathcal{C}_n) = \max_{P_M \in P(\mathcal{M}), \mathcal{S} \in \mathscr{S}} I_\mathcal{S}(M;\bm{Z})
 \end{equation}
 where $I_{\mathcal{S}}(M;Z)$ denotes the mutual information between $M \sim P_M$ and $\bm{Z} = \bm{x}(M,W,\mathcal{S})$. In turn, a family of codes $\{\mathcal{C}_n\}_{n=1}^{\infty}$ is said to be \textit{semantically-secret} if $\mathrm{Sem}(\mathcal{C}_n) = 2^{-\Omega(n)}$. We remark that this mutual-information based notation of SS is shown in \cite{Bellare2012} to be (asymptotically) equivalent to the operational definition of SS given in \cite{Goldwasser1984,Bellare1997}. Further, SS is a stronger notation of secrecy than strong secrecy.\footnote{A family of codes is said to achieve strong secrecy if $\lim_{n \rightarrow \infty} \max_{\mathcal{S} \in \mathscr{S}}I_{\mathcal{S}}(M;\bm{Z})=0$ where the message distribution is fixed s.t. $P_M \sim \mathrm{Unif}(\mathcal{M})$.}

 \textit{Reliability:} The (maximum) probability of decoding error is defined as
 $$P^{\mathrm{max}}_{\mathrm{error}}(\mathcal{C}_n) = \max_{m \in \mathcal{M}} \mathbb{P}\left( \widehat{M} \neq m | M =m \right)$$ where the probability is taken w.r.t. the distribution of Alice's key and the \textit{worst-case} distribution of the adversary's bit read/flip locations. A family of codes $\{\mathcal{C}_n\}_{n=1}^{\infty}$ is said to be \textit{reliable} if for any $\delta > 0$,
 $P_{\mathrm{error}}(\mathcal{C}_n) \leq \delta$ for large enough $n$. 
 
 \textit{SS Capacity:} The rate $R>0$ is said to be achievable over the $(p,r)$-AWTC II if there exists a family of codes $\{\mathcal{C}_n\}_{n=1}^{\infty}$ (where for each $n$, $\mathcal{C}_n$  is an $[n,Rn,R'n]$ code for some $R'\geq 0$) that is both semantically-secret and reliable. The SS capacity $C(p,r)$ is the supremum of rates achievable over the $(p,r)$-AWTC II.

 \subsection{Results} \label{sec:results}

 Our first result is on the necessity of non-linear codes for achieving the SS capacity. We say that a $[n,Rn,R'n]$ code $\mathcal{C}_n$ is \textit{linear}\footnote{Examples of linear codes in the wiretap setting include Ozarow's and Wyner's linear coset coding scheme \cite{Ozarow1986} and some polar code and LDPC code based schemes (e.g., \cite{Mahdavifar2011}).} if there exists a generator matrix $G \in \{0,1\}^{(R+R')n \times n}$ such that the codeword corresponding to any message $m \in \mathcal{M} \triangleq \{0,1\}^{Rn}$ and key $w \in \mathcal{W} \triangleq \{0,1\}^{R'n}$ is $\bm{x}(m,w) = \begin{bmatrix} m & w \end{bmatrix}G$. A corollary of the following Theorem is that for any $r \in (0,1]$ and $p=0$ (i.e., the channel to Bob is noiseless), linear codes cannot achieve SS capacity $C(0,r) = \max \{ 1 - r,0 \}$.

 \begin{theorem} \label{thm:linear_codes}
 Let $p =0$, $r \in (0,1]$, $R > \max\{1-2r,0\}$ and $R' \in [0,1-R]$. For large enough $n$, every linear $[n,Rn,R'n]$ code $\mathcal{C}_n$ has either semantic leakage $\mathrm{Sem}(\mathcal{C}_n) \geq 1$ or probability of error $P_{\mathrm{error}}(\mathcal{C}_n) \geq 1/2$ over the $(0,r)$-AWTC II.
 \end{theorem}

 \begin{remark}
 Theorem \ref{thm:linear_codes} can be extended to non-zero values of $p$. In particular, together with the lower bound (\ref{eq:SSC_ulb}), Theorem \ref{thm:linear_codes} implies that linear codes cannot achieve $C(p,r)$ for either any $p \in [0,1/2)$ and $r\in (0,1/2]$ such that $H_2(p)<r$, or any $p \in [0,1/2]$ and $r \in [1/2,1]$, except for the trivial case when $C(p,r)$ is $0$.
 \end{remark}
 
 A proof of Theorem \ref{thm:linear_codes} is given in Section \ref{sec:linear_codes_proof}, which involves a specific construction of the adversary's coordinates $\mathcal{S}$ together with the Plotkin bound to upper bound the minimum distance of a code. We remark that tighter distance bounds can be used in place of the Plotkin bound. For instance, if one uses the Elias-Bassalygo bound \cite{Bassalygo1965,Shannon1967_I,Shannon1967_II}, the rate lower bound in Theorem \ref{thm:linear_codes} can be tightened to $R > \max\{1 - H\left( \frac{1- \sqrt{1-2r}}{2}\right),0 \}$. All bounds discussed thus far are plotted in Fig. \ref{fig:sem_cap}.

  \begin{figure}[t]
    \includegraphics[width=\columnwidth]{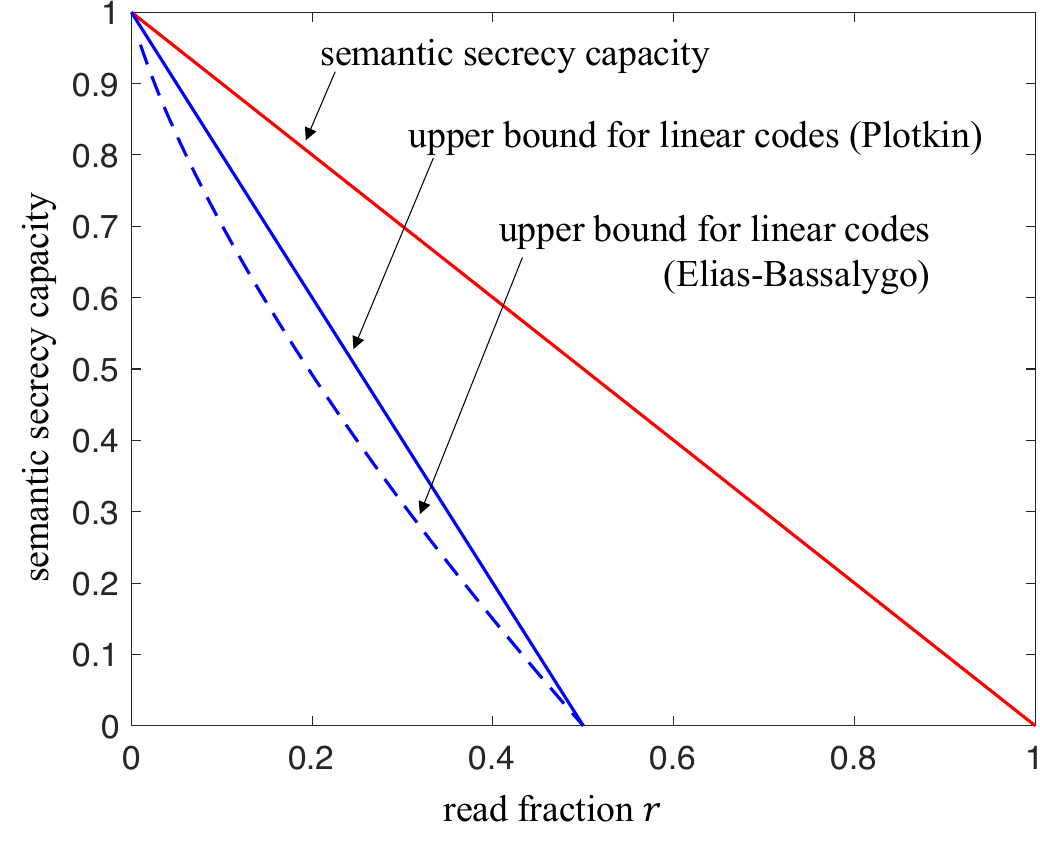}
    \caption{Bounds on linear code performance compared to the semantic secrecy capacity of the $(p,r)$-AWTC II when the channel from Alice to Bob is noiseless (i.e., $p=0$).}
    \label{fig:sem_cap}
    \end{figure}

  In light of Theorem \ref{thm:linear_codes}, non-linear codes must be considered to achieve the lower bound (\ref{eq:SSC_ulb}) for at least some values of $p \in [0,1/2]$ and $r \in [0,1]$. We turn now to non-linear codes.

 \begin{definition}[Pseudolinear Code] \label{def:plc}
 For $R\in(0,1]$, $R' \in [0,1-R]$ and positive integers $n$ and $k$, let $H$ be the parity check matrix of any binary linear code with the following parameters:
 \begin{itemize}
 \item blocklength $2^{(R+R')n}-1$ \item dimension $2^{(R+R')n}-1-\ell$ for some $\ell = O(k(R+R')n)$
 \item minimum distance at least $k+1$.
 \end{itemize}
 An $[n,Rn,R'n,k]$ psuedolinear code $\mathcal{C}_n$ is any $[n,Rn,R'n]$ code that satisfies the following two step encoding process. First, a message-key pair $(m,w) \in \mathcal{M} \times \mathcal{W}$ is mapped to the row of $H^T$ indexed by $(m,w)$, which we denote as $\bm{h}(m,w)$.\footnote{To account for the message-key pair $(0,0)$, we define $\bm{h}(0,0)$ to be the all zeros vector.} Second, $\bm{h}(m,w)$ is linearly mapped to an $n$-bit codeword by some ``generator'' matrix $G \in \{0,1\}^{\ell \times n}$, i.e., $$\bm{x}(m,w) = \bm{h}(m,w) G.$$
 Thus, the non-linearity of $\mathcal{C}_n$ is confined to the first stage of encoding.
 \end{definition}

 Towards the goal of derandomizing the random code of \cite{Wang2016}, pseudolinear codes have the following three attractive properties \cite{Guruswami2001_2}.\footnote{See \cite{Guruswami2001} for further discussion of pseudolinear codes.} First, a pseudolinear code has a succinct representation as only $\ell n = O(k(R+R')n^2)$ bits are needed to describe the generator matrix. Second, encoding is computationally efficient if $\bm{h}(m,w)$ can be obtained in time polynomial in $n$ for each $(m,w) \in \mathcal{M}\times\mathcal{W}$. For instance, we can let $H$ be the parity check matrix of a binary Bose–Chaudhuri–Hocquenghem (BCH) code of design distance $k+1$, in which case $H$ has an explicit representation and $\bm{h}(m,w)$ can be efficiently obtained by computing powers of a primitive $(2^{(R+R')n}-1)$-th root of unity from the extension field $\mathrm{GF}(2^{(R+R')n})$, e.g., see \cite{Macwilliams1977}. 
 
 Third, if we consider a \textit{random} pseudolinear code by choosing the generator matrix $G$ at random while fixing the parity check matrix $H$, then the codewords of the random code are uniformly distributed in $\{0,1\}^n$ and \textit{$k$-wise independent}, i.e., any subset of codewords of size $k$ are mutually independent.\footnote{In contrast, random linear codes have codewords that are pair-wise (i.e., 2-wise) independent in non-trivial cases.} This final property is the key to showing that pseudolinear codes achieve the best-known lower bound of $C(p,r)$.
 
 \begin{theorem} \label{thm:secret}
 Let $p \in [0,1/2]$ and $r \in [0,1]$ such that $1-H_2(p)-r$ is positive. For any $R < 1- H_2(p) - r$ and for large enough (but fixed) $k$, there exists a family pseudolinear codes $\{\mathcal{C}_n\}_{n=1}^\infty$  (where for $n\geq 1$, $\mathcal{C}_n$ is an $[n,Rn,R'n,k]$ pseudolinear code for some $R'\geq 0$) that is both reliable and semantically-secret.
 \end{theorem}

 A proof of Theorem \ref{thm:secret} is provided in Section \ref{sec:secret_proof}. The key technical tool in the proof is a new version of Wyner's soft-covering lemma  which holds for codes with $k$-wise independent codeword. However, our version differs significantly from Wyner's \cite[Theorem 6.3]{Wyner1975}, which we state and prove in Section \ref{sec:soft_covering}. 
 
 Our version is closest to (and proved similarly to) the soft-covering lemma of Goldfeld, Cuff and Permuter \cite{Goldfeld2016}, which roughly states that if the key rate $R'$ is larger than the mutual information between Alice's channel input and the adversary's observation, then a random code with mutually independent codewords satisfies an exponential number of secrecy constraints with probability at least $1-2^{-2^{\Omega(n)}}$. Here, the double-exponential probability bound is important as it allows one to take a union bound over an exponential number of events. Our version of the lemma states that when we restrict the random code to a $k$-wise independent distribution, the same constraints hold with probability at least $1-2^{-k \Omega(n)}$. Critically, while our probability bound tends to $1$ more slowly than double-exponentially, it remains fast enough to take a union bound over an exponential number of events when $k$ is large enough.

 \subsection{Related Work} \label{sec:related_work}

 \textit{Linear Codes and Semantic-Secrecy:} Recall that Theorem \ref{thm:linear_codes} states that linear codes cannot achieve the SS capacity for the (noiseless) $(0,r)$-AWTC II for any $r \in (0,1]$. Prior to this work, some special classes of linear codes were known to not achieve the SS capacity. In particular, Ozarow's and Wyner's linear coset coding scheme \cite{Ozarow1986} does not achieve SS capacity of the $(0,r)$-AWTC II for any $r \in (0,1]$. \ifthenelse{\boolean{extend_v}}{We provide a proof of this result in Appendix \ref{sec:LCCS}.}{We provide a proof of this result in the extended paper \cite{Ruzomberka2023_2}.} We remark that the necessity of non-linear codes for achieving the secrecy capacity is a product of the \textit{joint consideration} of the semantic secrecy metric and the type II property of the wiretap channel. In contrast, linear codes are sufficient to achieve the \textit{weak} secrecy capacity over the noiseless WTC II \cite{Ozarow1986}. Furthermore, linear codes are sufficient to achieve both the weak and strong secrecy capacity of the noisy (but non-adversarial) WTC I \cite{Mahdavifar2011}.


 \textit{Code Constructions:} Explicit (and efficient) constructions that achieve the best known lower bound of the $(p,r)$-AWTC II are not known in general, except for the special cases of $p=0$ \cite{Cheraghchi2012,Chou2022} and $r = 0$ \cite{Guruswami2016}. In the general case, one promising approach is use modular constructions, which combine an existing error-control code with an invertible extractor \cite{Bellare2012,Sharifian2018,Chou2022} or algebraic manipulation detection code \cite{Wang2016_2}. However, constructing binary error-control codes that are both efficiently encodable/decodable and achieve the (reliability) capacity of the $(p,r)$-AWTC is an open problem. In contrast to the above modular constructions, pseudolinear codes offer a non-modular approach. Recently, random (and thus non-explicit) pseudolinear codes were shown to achieve the (reliability) capacity of the $(p,r)$-AWTC II \cite{Ruzomberka2023}.

 \section{Proof of Theorem \ref{thm:linear_codes}} \label{sec:linear_codes_proof}

  \textit{Notation:} For message rate $R>0$, key rate $R'\in [0,1-R]$, and blocklength $n \geq 1$ define $\mathcal{M} \triangleq \{0,1\}^{Rn}$ and $\mathcal{W} \triangleq \{0,1\}^{R'n}$. For an $[n,Rn,R'n]$ linear code $\mathcal{C}_n$, let $G$ denote the $(R+R')n \times n$ generator matrix of $\mathcal{C}_n$, which can be partitioned such that $G = \begin{bmatrix} G_M \\ G_W \end{bmatrix}$ where $G_M \in \{0,1\}^{Rn \times n}$ and $G_W \in \{0,1\}^{R'n \times n}$. In turn, the codeword corresponding to message $m \in \mathcal{M}$ and key $w \in \mathcal{W}$ is $\bm{x}(m,w) = m G_M + w G_W$. For a coordinate set $\mathcal{S} \in \mathscr{S}$, let the matrices $G_M(\mathcal{S})$ and $G_W(\mathcal{S})$ denote the columns of $G_M$ and $G_W$ indexed by $\mathcal{S}$, respectively. Using this notation, if Alice transmits codeword $\bm{x}(m,w)$ then the adversary observes $\bm{z} = m G_M(\mathcal{S}) + w G_W(\mathcal{S})$.

  \textit{Preliminaries:} Let $\mathcal{C}_n$ be an $[n,Rn,R'n]$ linear code with generator matrix $G$. We make the following assumption.

  \begin{assumption} \label{ass:rank}
 Without loss of generality (w.l.o.g.), we assume that $G$ is full rank, i.e., $\mathrm{rank}(G) = (R+R')n$.
 \end{assumption}

 The claim being w.l.o.g. is roughly as follows: if $G$ is not full rank, then either $P^{\mathrm{max}}_{\mathrm{error}}(\mathcal{C}_n) \geq 1/2$ or both $\mathcal{W}$ and $G$ can be replaced with a smaller key set and full rank generator matrix, respectively, without changing the code. \ifthenelse{\boolean{extend_v}}{A detailed discussion is provided in Appendix \ref{sec:assrank_proof}.}{A detailed discussion is provided in the extended paper \cite{Ruzomberka2023_2}.} We remark that following Assumption \ref{ass:rank}, we have that $\mathrm{rank}(G_M) = Rn$ and $\mathrm{rank}(G_W)= R'n$. 

 Before proving the converse result (Theorem \ref{thm:linear_codes}), we state a few preliminary results relating the semantic leakage to the rank of $G_M(\mathcal{S})$ and $G_W(\mathcal{S})$ for $\mathcal{S} \in \mathscr{S}$. For a code $\mathcal{C}_n$ and coordinate set $\mathcal{S} \in \mathscr{Z}$, we denote the mutual information between $M$ and $\bm{Z}$ as $I_{\mathcal{S}}(M;\bm{Z})$ (where the dependency on $\mathcal{C}_n$ is implied).

 \begin{lemma} \label{thm:rank}
 For $\mathcal{S} \in \mathscr{S}$ and $M$ uniformly distributed over $\mathcal{M}$,
 $$I_{\mathcal{S}}(M;\bm{Z}) = \mathrm{rank}\left( G(S) \right) - \mathrm{rank} \left(G_W(\mathcal{S}) \right).$$ 
 \end{lemma}

 \begin{proof}[Proof of Lemma \ref{thm:rank}]
 Let $\mathcal{S} \in \mathscr{S}$. We first characterize the joint PMF of $M$, $W$ and $\bm{Z}$, which we denote as $P_{M,W,\bm{Z}}$. We drop the subscripts from the PMF $P_{M,W,\bm{Z}}$ and its marginal PMFs when the meaning is clear from the use of the realization variables $m$, $w$ and $\bm{z}$. 
 
 For $\bm{z} \in \{0,1\}^{rn}$ and $m \in \mathcal{M}$, we have that
 \begin{align}
 P(\bm{z}|m) &= \sum_{w \in \mathcal{W}} P(\bm{z},w|m)  \stackrel{(\text{a})}{=} \sum_{w \in \mathcal{W}} P(\bm{z}|m,w) P(w) \nonumber \\
 & \stackrel{(\text{b})}{=} T_{m,\bm{z}} 2^{-R'n} \label{eq:rank_1}
 \end{align}
 where (a) follows from the independence of $M$ and $W$, (b) follows from $W \sim \mathrm{Unif}(\mathcal{W})$, and where $T_{m,\bm{z}} \triangleq \sum_{w \in \mathcal{W}} \mathds{1} \{ \bm{z} = m G_M(\mathcal{S}) + w G_W(\mathcal{S})\}$. 
 
 To simplify (\ref{eq:rank_1}), define $$\mathcal{T} \triangleq \left\{(m',\bm{z}') \in \mathcal{M} \times \{0,1\}^{rn}: T_{m',\bm{z}'} \geq 1 \right\}$$ and suppose that $(m,\bm{z}) \in \mathcal{T}$. By definition, there exists an $w \in \mathcal{W}$ such that $w G_W(\mathcal{S}) = m G_M(\mathcal{S}) + \bm{z}$. In turn, since the mapping $G_W(\mathcal{S}):\mathcal{W} \rightarrow \{0,1\}^{rn}$ is a linear transformation, there must be $2^{\mathrm{nullity}(G_W(\mathcal{S}))}$ number of $w \in \mathcal{W}$ such that $w G_W(\mathcal{S}) = m G_M(\mathcal{S}) + \bm{z}$ where $\mathrm{nullity}(G_W(\mathcal{S}))$ is the dimension of the null space of $G_W(\mathcal{S})$. By the rank-nullity theorem \cite[Theorem 2]{Hoffman1971}, $2^{\mathrm{nullity}(G_W(\mathcal{S}))} = 2^{\mathrm{dim}(\mathcal{W})-\mathrm{rank}(G_W(\mathcal{S}))} = 2^{R' n-\mathrm{rank}(G_W(\mathcal{S}))}$. In turn,
 \begin{equation} \nonumber
 T_{m,\bm{z}} = 
 \begin{cases}
 2^{R'n - \mathrm{rank}(G_W(\mathcal{S})) }, & (m,\bm{z}) \in \mathcal{T} \\
 0, & (m,\bm{z}) \not\in \mathcal{T},
 \end{cases}
 \end{equation}
 and in turn, following (\ref{eq:rank_1}),
 \begin{equation} \label{eq:rank_2}
 P(\bm{z}|m) = 
 \begin{cases}
 2^{ -\mathrm{rank}(G_W(\mathcal{S}))}, & (m,\bm{z}) \in \mathcal{T} \\
 0, & (m,\bm{z}) \not\in \mathcal{T}.
 \end{cases}
 \end{equation}
 Repeating the above approach for the PMF of $\bm{Z}$, one can show using the assumption that $m$ is uniformly distributed over $\mathcal{M} = \{0,1\}^{Rn}$ that
 \begin{equation} \label{eq:rank_3}
 \begin{aligned}
  P(\bm{z}) =  \begin{cases}
 2^{ -\mathrm{rank}(G(\mathcal{S}))}, & \exists m \in \mathcal{M} \text{ s.t. } (m,\bm{z}) \in \mathcal{T} \\
 0, & \forall m \in \mathcal{M}, \text{ } (m,\bm{z}) \not\in \mathcal{T}.
 \end{cases}
 \end{aligned}
 \end{equation}

 Using the above PMFs, we evaluate the mutual information between $M$ and $\bm{Z}$:
 \begin{align}
 I_{\mathcal{S}}(M;\bm{Z}) &\triangleq  \sum_{m \in \mathcal{M}} \sum_{\bm{z} \in \{0,1\}^{rn}} P(m,\bm{z}) \log_2 \frac{P(\bm{z}|m)}{P(\bm{z})} \nonumber \\
 & \stackrel{(\text{c})}{=}  \sum_{(m,\bm{z}) \in \mathcal{T}} P(m,\bm{z}) \log_2 2^{\mathrm{rank}(G(\mathcal{S})) - \mathrm{rank}(G_W(\mathcal{S}))} \nonumber \\
 & \stackrel{(\text{d})}{=} \mathrm{rank}\left(G(\mathcal{S})\right) - \mathrm{rank}\left(G_W(\mathcal{S})\right). \nonumber
 \end{align}
 where (c) follows from (\ref{eq:rank_2}), (\ref{eq:rank_3}), and $P(m,\bm{z}) = 0$ $\forall (m,\bm{z}) \not\in \mathcal{T}$, and (d) follows from $\sum_{(m,\bm{z}) \in \mathcal{T}} P(m,\bm{z}) = 1$.
 \end{proof}

 \begin{corollaryLemma} \label{thm:Rlr}
 If $R'+R \leq r$, then $\lim_{n \rightarrow \infty} \mathrm{Sem}(\mathcal{C}_n) = \infty$.
 \end{corollaryLemma}

 \begin{proof}[Proof of Corollary \ref{thm:Rlr}]
 Suppose that $M$ is uniformly distributed and that $R+R' \leq r$. Recall that $G$ has rank $(R+R')n$ (c.f. Assumption \ref{ass:rank}). Since $R+R' \leq r$, there exists a $\mathcal{S} \in \mathscr{S}$ such that $\mathrm{rank}(G(\mathcal{S})) = \mathrm{rank}(G) = (R+R') n$. Let $\mathcal{S}$ be this coordinate set. It follows that $\mathrm{rank}(G_W(\mathcal{S})) = R'n$, and in turn, $I_\mathcal{S}(M;\bm{Z}) = Rn$ following Lemma \ref{thm:rank}. In conclusion, $\mathrm{Sem}(\mathcal{C}_n) \geq Rn$.
 \end{proof}

  For the converse analysis, we will need the following version of the Plotkin bound \cite{Plotkin1960}.
 \begin{lemma}[{Extended Plotkin bound \cite{Rudra2007}}] \label{thm:Plotkin}
 Suppose that $\Psi$ is an $[n,Rn]$ code (not necessarily linear) with minimum distance $d_{\mathrm{min}} \in (0,n/2]$. Then for $\delta \triangleq d_{\mathrm{min}}/n$, $$R \leq 1 - 2 \delta + o(1)$$ where the $o(1)$ term tends to $0$ as $n$ tends to infinity.
 \end{lemma}

 \textit{Converse (Proof of Theorem \ref{thm:linear_codes}) Setup:} Set $p=0$ and let $r \in [0,1]$. For any $\epsilon > 0$, let $R = \max \{1 - 2r,0 \} + \epsilon$ and let $R' \in [0,1-R]$ such that $R+R'>r$ (c.f. Corollary \ref{thm:Rlr}). In turn, we let $\mathcal{C}_n$ be an $[n,Rn,R'n]$ linear code with generator matrix $G$. W.l.o.g., we assume that $G$ is full rank (c.f. Assumption \ref{ass:rank}).

 \textit{Converse Attack:} The adversary orchestrates it attack in two steps. First, the adversary chooses an index set $\mathcal{V}\subseteq [n]$ of size $(R+R')n$ such that all columns of $G(\mathcal{V})$ are linearly independent. Note that such a set exists following our assumption that $G$ is rank $(R+R')n$. Second, the adversary chooses a coordinate set $\mathcal{S}^* \in \mathscr{S}$ to be a subset of $\mathcal{V}$ that minimizes the rank of $G_W(\mathcal{S}^*)$. Once Alice transmits her codeword $\bm{x}(M,W)$, the adversary reads the codeword bits $\bm{Z} = \bm{x}(M,W,\mathcal{S}^*)$ corresponding to the coordinates $\mathcal{S}^*$ with corresponding mutual information $I_{\mathcal{S}^*}(M;\bm{Z})$.
 
 \textit{Converse Analysis:} The goal of the converse analysis is to show that $I_{\mathcal{S}^*}(M;\bm{Z}) \geq 1$. We remark that $\mathcal{S}^*$ is a strict subset of $\mathcal{V}$ following the inequality $r<R+R'$. This fact together with the fact that all $|\mathcal{V}|$ column of $G(\mathcal{V})$ are linearly independent implies that the rank of $G(\mathcal{S}^*)$ is $rn$. In turn, following Lemma \ref{thm:rank}, 
 \begin{align}
 I_{\mathcal{S}^*}(M;\bm{Z}) = rn - \mathrm{rank}(G_W(\mathcal{S}^*)). \label{eq:suff_I}
 \end{align}
 In the converse analysis, we show that $rn - \mathrm{rank}(G_W(\mathcal{S}^*)) \geq 1$.

 We proceed with the following dual code perspective. Consider $G_W(\mathcal{V})$ as the $R'n \times (R+R')n$ generator matrix of some $[(R+R')n,R'n]$ linear code $\Psi$. In turn, let $G_W^{\perp}(\mathcal{V})$ denote the $Rn \times (R+R')n$ generator matrix of the $[(R+R')n,Rn]$ dual code $\Psi^{\perp}$ of $\Psi$. By definition, $G_W(\mathcal{V})$ is the parity check matrix corresponding to the generator matrix $G_W^{\perp}(\mathcal{V})$. Let $d^{\perp}_{\mathrm{min}}$ denote the minimum distance of $\Psi^{\perp}$. By the definition of the parity check matrix (e.g., see \cite{Macwilliams1977}), there exists $d^{\perp}_{\mathrm{min}}$ linearly dependent columns of the parity check matrix $G_W(\mathcal{V})$.
 Hence, if $d^{\perp}_{\mathrm{min}} \leq rn$, then the adversary's choice of $\mathcal{S}^*$ contains the indices of these $d^{\perp}_{\mathrm{min}}$ linearly dependent columns of $G_W$, i.e, the rank of $G_W(\mathcal{S}^*)$ is bounded above by $rn - 1$. In turn, $I_{\mathcal{S}^*}(M;Z) \geq 1$ via (\ref{eq:suff_I}). To complete the proof, we show that $d^{\perp}_{\mathrm{min}} \leq rn$.
 
 Applying the Plotkin bound (Lemma \ref{thm:Plotkin}) to the dual code $\Psi^{\perp}$, we have that
 \begin{equation} \label{eq:plotkin}
 \frac{R}{R+R'} \leq 1 - 2 \delta^{\perp} + o(1) 
 \end{equation}
 for the distance parameter $\delta^{\perp} \triangleq \frac{d^{\perp}_{\mathrm{min}}}{(R+R')n}$ and where the $o(1)$ term tends to $0$ as $n$ tends to infinity. In turn, for large enough $n$,
 \begin{align}
 d^{\perp}_{\mathrm{min}} & \stackrel{(d)}{\leq} \frac{R'n}{2} + o(n) \nonumber \\
 & \stackrel{(e)}{\leq} \frac{2r - \epsilon}{2} + o(n) \nonumber \\
 & \stackrel{(f)}{<} rn \nonumber
 \end{align}
 where (d) follows from a rearrangement of (\ref{eq:plotkin}), (e) follows from the setting of rate $R = \max \{1 - 2r,0 \} + \epsilon$ and the trivial inequalities $R+R' \leq 1$ and $\max \{1-2r,0 \} \geq 1- 2r$, and (f) follows for large enough $n$. In conclusion, for large enough $n$, $I_{\mathcal{S}^*}(M;\bm{Z}) \geq 1$ and thus $\mathrm{Sem}(\mathcal{C}_n) \geq 1$.

 \section{A Soft-Covering Lemma for $k$-wise Independent Codewords} \label{sec:soft_covering}

 \textit{Notation:} In this section only, we consider a more general code model than that introduced in Section \ref{sec:setup}. For an alphabet $\mathcal{U}$ which is not necessarily binary, a blocklength $n$ and a (private) key rate $R' > 0$, we define an $[n,R'n]$ code $\mathcal{C}_n$ as a subset of $\mathcal{U}^n$ of size $|\mathcal{C}_n| = 2^{R'n}$. We will often describe $\mathcal{C}_n$ by its set of codewords $\{ \bm{u}(w,\mathcal{C}_n) \}_{w \in \mathcal{W}}$ for a key set $\mathcal{W} = [2^{R'n}]$.
 
We introduce the soft-covering problem, depicted in Fig. \ref{fig:coding_prob_0}. The problem setup is as follows. For a blocklength $n \geq 1$, let $\mathcal{C}_n = \{\bm{u}(w,\mathcal{C}_n)\}_{w \in \mathcal{W}}$ be an $[n,R'n]$ code. Given a finite input alphabet $\mathcal{U}$, an input distribution $Q_U$, a finite output alphabet $\mathcal{V}$ and channel $Q_{V|U}$, consider the PMFs induced on the output sequence $\bm{V} \in \mathcal{V}^n$ when an input sequence $\bm{U} \in \mathcal{U}^n$ is sent through the $n$-shot memoryless channel $Q_{V|U}^n$: for $\bm{v} \in \mathcal{V}^n$,
\begin{enumerate}
\item The PMF of $\bm{V}$ when $\bm{U}$ is drawn randomly from $Q^n_U$, i.e.,
$$Q_{\bm{V}}(\bm{v}) = Q_{V}^n(\bm{v}) = \sum_{\bm{u} \in \mathcal{U}} Q^n_{V|U}(\bm{v}|\bm{u}) Q_{U}^n(\bm{u}).$$
\item The PMF of $\bm{V}$ when  $\bm{U}$ is the codeword $\bm{u}(W,\mathcal{C}_n)$ for $W \sim \mathrm{Unif}(\mathcal{W})$, i.e., 
 \begin{equation} \label{eq:PVC}
 P^{(\mathcal{C}_n)}_{\bm{V}}(\bm{v}) \triangleq \sum_{w \in \mathcal{W}} Q^n_{V|U}(\bm{v}|\bm{u}(w,\mathcal{C}_n)) 2^{-Rn}.
 \end{equation}

\end{enumerate}

The soft-covering problem asks how to design a code $\mathcal{C}_n$ such that the induced PMF $\mathcal{P}^{(\mathcal{C}_n)}_{\bm{V}}$ is approximately $Q_{V}^n$ in the limit as $n$ tends to infinity. The following lemma states that if $R' > I(U;V)$, then for any integer $k$ large enough a \textit{random $[n,R'n]$ code} $\mathscr{C}_n$ with $k$-wise independent codewords each drawn from distribution $Q_U^n$ results in $P^{(\mathscr{C}_n)}_{\bm{V}} \approx Q^n_{V}$ for large enough $n$. Recall that we denote random codes with script typeface (e.g., $\mathscr{C}_n$) and we denote realizations of random codes with calligraphic typeface (e.g., $\mathcal{C}_n$).

   \begin{figure}[t]
  \includegraphics[width=\columnwidth]{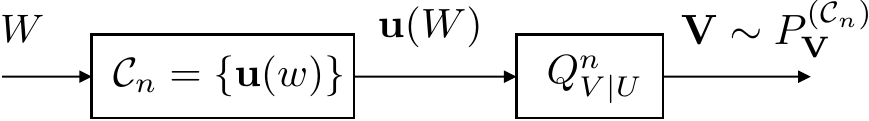}
  \caption{The soft-covering problem: the goal is to design the code $\mathcal{C}_n$ to make $P^{(\mathcal{C}_n)}_{\bm{V}} \approx Q^n_{V}$.}
  \label{fig:coding_prob_0}
  \end{figure}

 \begin{lemma}[Soft-covering lemma for $k$-wise independent codewords] \label{thm:sscl_LI}
 Suppose that the random code $\mathscr{C}_n$ has $k$-wise independent codewords for some even integer $k\geq 4$, each drawn from a PMF $Q_U^n$ for finite $\mathcal{U}$. Let $Q_{V|U}$ be any conditional PMF where $|\mathcal{V}|$ is finite and let $R' > I(U;V)$. There exists some $\gamma_0 >0$ and $\gamma_1 >0$ that depend only on $R'$ and $I(U;V)$ such that for large enough $n$
 $$\mathbb{P}_{\mathscr{C}_n} \left( D\left(P_{\bm{V}}^{(\mathscr{C}_n)} \Big|\Big| Q^n_{V} \right) > 2^{-\gamma_1 n} \right) \leq 2^{(-k \gamma_0 + \log_2 |\mathcal{V}|) n}$$
 where we recall that $D$ is the relative entropy.
 \end{lemma}

 \subsection{Overview of Proof of Lemma \ref{thm:sscl_LI}}

 \textit{Setup:} Let the blocklength $n \geq 1$ and key rate $R' > I(U;V)$, and let $k$ be a positive integer that will be set later. In turn, let $\mathscr{C}_n$ be a random $[n,R'n]$ code drawn from any distribution that has $k$-wise independent codewords each with marginal PMF $Q^n_U$.
 
 The proof of Lemma \ref{thm:sscl_LI} follows a two step approach. In the first step, the proof closely follows the proof outline of \cite{Goldfeld2016} in which we construct an upper bound on the relative entropy $D(P_{\bm{V}}^{(\mathscr{C}_n)} || Q^n_{V})$ based on a typical set construction of $n$-symbol sequences. In the second step, the proof diverges from \cite{Goldfeld2016} to analyze how the relative entropy upper bound concentrates. This second step uses the $k$-wise independent property of the random code $\mathscr{C}_n$.
 
 Define the information density of a scalar pair $(u,v) \in \mathcal{U} \times \mathcal{V}$ as
 $ i_{Q_{U,V}}(u;v) \triangleq \log_2 \frac{Q_{V|U}(v|u)}{Q_{V}(v)}.$ In turn, define the information density of an $n$-symbol sequence pair $(\bm{u},\bm{v}) \in \mathcal{U}^n \times \mathcal{V}^n$,
 $$ i_{Q^n_{U,V}}(\bm{u};\bm{v}) \triangleq \sum_{j=1}^n i_{Q_{U,V}}(u_j;v_j).$$
 For $\epsilon > 0$, define a typical set of $n$-symbol sequence pairs 
 $$\mathcal{A}_{\epsilon} \triangleq \left\{ (\bm{u},\bm{v}) \in \mathcal{U}^n \times \mathcal{V}^n: i_{Q^n_{U,V}}(\bm{u};\bm{v}) < \left(I(U;V)+\epsilon \right)n \right\}.$$ Recall that for an $[n,R'n]$ code $\mathcal{C}_n$, the PMF $P^{(\mathcal{C}_n)}_{\bm{V}}$ is the PMF of $\bm{V}$ when $\bm{U}$ is a codeword drawn from the code $\mathcal{C}_n$ (c.f. (\ref{eq:PVC})). We split $P^{(\mathcal{C}_n)}_{\bm{V}}$ into two terms based on the typical set $\mathcal{A}_\epsilon$: for $\bm{v} \in \mathcal{V}^n$, define
 \begin{align}
 &P^{(\mathcal{C}_n)}_{\bm{V},1}(\bm{v}) \triangleq \nonumber \\
 & \hspace{2em} 2^{-Rn} \sum_{w \in \mathcal{W}} Q^n_{V|U}(\bm{v}|\bm{u}(w,\mathcal{C}_n)) \mathds{1} \{(\bm{u}(w,\mathcal{C}_n),\bm{v}) \in \mathcal{A}_\epsilon\}, \nonumber
 \end{align}
 and define
 \begin{align}
 &P^{(\mathcal{C}_n)}_{\bm{V},2}(\bm{v}) \triangleq \nonumber \\
 & \hspace{2em} 2^{-Rn} \sum_{w \in \mathcal{W}} Q^n_{V|U}(\bm{v}|\bm{u}(w,\mathcal{C}_n)) \mathds{1} \{(\bm{u}(w,\mathcal{C}_n),\bm{v}) \not\in \mathcal{A}_\epsilon\}. \nonumber
 \end{align}
 By inspection, $P^{(\mathcal{C}_n)}_{\bm{V}} = P^{(\mathcal{C}_n)}_{\bm{V},1} + P^{(\mathcal{C}_n)}_{\bm{V},2}$; note that $P^{(\mathcal{C}_n)}_{\bm{V},1}$ and $P^{(\mathcal{C}_n)}_{\bm{V},2}$ may not be PMFs. We also define the ratios
 \begin{equation} \nonumber
 \Delta^{(\mathcal{C}_n)}_{\bm{V},1}(\bm{v}) \triangleq \frac{P^{(\mathcal{C}_n)}_{\bm{V},1}(\bm{v})}{Q^n_{V}(\bm{v})} \hspace{1em} \text{ and } \hspace{1em}
 \Delta^{(\mathcal{C}_n)}_{\bm{V},2}(\bm{v}) \triangleq \frac{P^{(\mathcal{C}_n)}_{\bm{V},2}(\bm{v})}{Q^n_{V}(\bm{v})}.
 \end{equation}
 We restate a result from \cite{Goldfeld2016} that bounds the relative entropy of $P^{(\mathcal{C}_n)}_{\bm{V}}$ and $Q^n_V$ in terms of the introduced quantities.

 \begin{lemma}[{\cite[Lemma 3]{Goldfeld2016}}] \label{thm:RE_UB}
 For every $[n,R'n]$ code $\mathcal{C}_n$,
 \begin{equation} \nonumber
 \begin{aligned}
 &D\left( P^{(\mathcal{C}_n)}_{\bm{V}} \Big|\Big| Q^n_{V} \right) \leq H_2 \left( \sum_{\bm{v} \in \mathcal{V}^n} P^{(\mathcal{C}_n)}_{\bm{V},2}(\bm{v}) \right) \\
 &\hspace{5em} + D\left( P^{(\mathcal{C}_n)}_{\bm{V},1} \Big|\Big| Q^n_{V} \right) + D\left( P^{(\mathcal{C}_n)}_{\bm{V},2} \Big|\Big| Q^n_{V} \right).
 \end{aligned}
 \end{equation}
 \end{lemma}

 We remark that the RHS of the inequality of Lemma \ref{thm:RE_UB} is well defined if we extend the definition of relative entropy $D(\cdot||\cdot)$ in the natural way to account for functions $P^{(\mathcal{C}_n)}_{\bm{V},1}$ and $P^{(\mathcal{C}_n)}_{\bm{V},2}$ which may not be PMFs. The following sufficient condition for Lemma \ref{thm:sscl_LI} follows from Lemma \ref{thm:RE_UB}.

 \begin{lemma}[Sufficient Condition for Lemma \ref{thm:sscl_LI}] \label{thm:suff_sscl_LI}
 Suppose that for some $\pi_0 \in [0,1]$ and with probability at least $1-\pi_0$ over the random code distribution, for some $\pi_1>0$
 \begin{equation} \label{eq:suff_1}
  \sum_{\bm{v} \in \mathcal{V}^n} P^{(\mathscr{C}_n)}_{\bm{V},2}(\bm{v}) < 2^{-\pi_1 n}
 \end{equation}
 and
 \begin{equation} \label{eq:suff_2}
 \Delta^{(\mathscr{C}_n)}_{\bm{V},1}(\bm{v}) < 1 + 2^{-\pi_1 n} \text{ for all } \bm{v} \in \mathcal{V}^n.
 \end{equation}
 Then 
 \begin{equation} \label{eq:suff_3}
 \mathbb{P}_{\mathscr{C}_n}\left( D\left( P^{(\mathscr{C}_n)}_{\bm{V}} \Big|\Big| Q^n_{V} \right) \geq q_n 2^{-\pi_1 n} \right) \leq \pi_0
 \end{equation}
 where $q_n = 2\log_2 e + \pi_1 n + n \log_2 \left( \max\limits_{v \in \mathrm{supp}(Q_V)} \frac{1}{Q_V(v)} \right)$.
 \end{lemma}

 \ifthenelse{\boolean{extend_v}}{
 \begin{proof}[Proof of Lemma \ref{thm:suff_sscl_LI}]
 Let $\pi_1>0$ and suppose that $\mathcal{C}_n$ is a realization of $\mathscr{C}_n$ such that both (\ref{eq:suff_1}) and  (\ref{eq:suff_2}) hold. We bound each of the 3 terms in the inequality of Lemma \ref{thm:RE_UB} using (\ref{eq:suff_1}) and  (\ref{eq:suff_2}). 
 
 Consider the first term. Following (\ref{eq:suff_1}) and the inequality\footnote{This inequality follows from an application of both the inequality $\frac{x}{1+x} \leq \ln(1+x)$ for $x>-1$ and the definition of $H_2(x)$.} $H_2(x) \leq x \log_2 \frac{e}{x}$ for $x \in [0,1]$, we have that 
 \begin{equation} \label{eq:suff_4}
 H_2\left( \sum_{\bm{v} \in \mathcal{V}^n} P^{(\mathcal{C}_n)}_{\bm{V},2}(\bm{v}) \right) \leq H_2(2^{-\pi_1 n}) < 2^{-\pi_1 n} (\log_2 e + \pi_1 n). 
 \end{equation}
 Moving on to the second term, following (\ref{eq:suff_2}) and the inequality $\log_2(1+x) \leq x \log_2 e$ for $x>0$, we have that 
 \begin{align}
 D(P^{(\mathcal{C}_n)}_{\bm{V},1} || Q_V^n) &\triangleq \sum_{\bm{v} \in \mathcal{V}^n} P^{(\mathcal{C}_n)}_{\bm{V},1}(\bm{v}) \log_2 \Delta^{(\mathcal{C}_n)}_{\bm{V},1} \nonumber \\
 &< \sum_{\bm{v}\in \mathcal{V}^n} P^{(\mathcal{C}_n)}_{\bm{V},1} \log_2 (1+2^{-\pi_1 n}) \nonumber \\
 & \leq \log_2(1+2^{-\pi_1 n}) \leq 2^{-\pi_1 n} \log_2 e. \label{eq:suff_5}
 \end{align}
 Moving to the last term, we will use the following inequality which uses the assumption that $|\mathcal{V}|$ is finite: $\Delta^{(\mathcal{C}_n)}_{\bm{V},2}(\bm{v}) \triangleq \frac{P^{(\mathcal{C}_n)}_{\bm{V},2}(\bm{v})}{Q^n_{V}(\bm{v})} \leq \max\limits_{\bm{v}' \in \mathrm{supp}(Q^n_V)} \frac{1}{Q^n(\bm{v}')} = (\max\limits_{v' \in \mathrm{supp}(Q_V)} \frac{1}{Q(v')})^n$ for all $\bm{v} \in \mathcal{V}^n$. Following this inequality and (\ref{eq:suff_1}), we have that 
 \begin{align}
 &D(P^{(\mathcal{C}_n)}_{\bm{V},2} || Q^n_{V})  \triangleq \sum_{\bm{v} \in \mathcal{V}^n} P^{(\mathcal{C}_n)}_{\bm{V},2}(\bm{v}) \log_2 \Delta^{(\mathcal{C}_n)}_{\bm{V},2} \nonumber \\
 & \hspace{2em} \leq \sum_{\bm{v} \in \mathcal{V}^n} P^{(\mathcal{C}_n)}_{\bm{V},2}(\bm{v}) n \log_2 \left( \max_{v' \in \mathrm{supp}(Q_V)} \frac{1}{Q_V(v')} \right)\nonumber \\
 & \hspace{2em} < 2^{-\pi_1 n} n \log_2 \left( \max_{v' \in \mathrm{supp}(Q_V)} \frac{1}{Q_V(v')} \right). \label{eq:suff_6}
 \end{align}
 Combining the bounds (\ref{eq:suff_4}), (\ref{eq:suff_5}) and (\ref{eq:suff_6}) together with Lemma Lemma \ref{thm:RE_UB}, the desired inequality (\ref{eq:suff_3}) immediately follows.
 \end{proof}
 }{A proof of Lemma \ref{thm:suff_sscl_LI} is available in the extended version \cite{Ruzomberka2023_2}.}

 In the remainder of the proof of Lemma \ref{thm:sscl_LI}, we apply the framework of the sufficient condition (Lemma \ref{eq:suff_2}) and show that inequalities (\ref{eq:suff_1}) and (\ref{eq:suff_2}) hold with probability $1-\pi_0$ over the distribution of $\mathscr{C}_n$ for a value $\pi_0 = 2^{-k\Omega(n) + n \log_2|\mathcal{V}|}$ and some $\pi_1 > 0$. As the primary technical tools of the proof, we use the concentration inequalities of Schmidt, Siegel and Srinivasan \cite{Schmidt1995} and Bellare and Rompel \cite{Bellare1994} for sums of $k$-wise independent random variables. 

 \subsection{Proof of Lemma \ref{thm:sscl_LI}}
 First, we show that inequality (\ref{eq:suff_1}) holds with high probability over the random code $\mathscr{C}_n$ for some $\pi_1>0$. Consider the quantity 
 \begin{align}
 &\sum_{\bm{v} \in \mathcal{V}^n} P^{(\mathscr{C}_n)}_{\bm{V},2}(\bm{v}) \nonumber \\
 &= \Scale[0.93]{\sum\limits_{w \in \mathcal{W}} 2^{-R'n} \sum\limits_{\bm{v} \in \mathcal{V}^n} Q^n_{V|U}(\bm{v}|\bm{U}(w,\mathscr{C}_n)) \mathds{1} \left\{ \left(\bm{U}(w,\mathscr{C}_n),\bm{v}\right) \not\in \mathcal{A}_\epsilon \right\}} \nonumber \\
 &= \Scale[0.93]{\sum\limits_{w \in \mathcal{W}} 2^{-R'n} \mathbb{P}_{\bm{V} \sim Q^n_{V|U}} \left( (\bm{U}(w,\mathscr{C}_n),\bm{V}) \not\in \mathcal{A}_\epsilon \Big| \bm{U} = \bm{U}(w,\mathscr{C}_n)\right)} \label{eq:P2_exp}
 \end{align}
 Note that (\ref{eq:P2_exp}) is a sum of $|\mathcal{W}|=2^{R'n}$ $k$-wise-independent terms following that the codewords of $\mathscr{C}_n$ are $k$-wise independent. 

 For $w \in \mathcal{W}$, the expectation of the $w^{\text{th}}$ term in the sum of (\ref{eq:P2_exp}) is
 \begin{align}
 &2^{-R'n} \mathbb{E}_{\mathscr{C}_n} \mathbb{P}_{\bm{V} \sim Q^n_{V|U}} \left( (\bm{U}(w,\mathscr{C}_n),\bm{V}) \not\in \mathcal{A}_\epsilon \Big| \bm{U} = \bm{U}(w,\mathscr{C}_n)\right) \nonumber \\
 & \stackrel{(\text{a})}{=} 2^{-R'n} \mathbb{P}_{(\bm{U},\bm{V}) \sim Q^n_{U,V}} \left( (\bm{U},\bm{V}) \not\in \mathcal{A}_\epsilon \right) \nonumber \\
 & \stackrel{(\text{b})}{=} 2^{-R'n} \mathbb{P}_{(\bm{U},\bm{V}) \sim Q^n_{U,V}} \left(i_{Q^n_{U,V}}(\bm{U};\bm{V}) \geq (I(U;V)+\epsilon)n \right) \nonumber \\
 & \stackrel{(\text{c})}{=} 2^{-R'n} \mathbb{P}_{(\bm{U},\bm{V}) \sim Q^n_{U,V}} \left( 2^{\lambda i_{Q^n_{U,V}}(\bm{U};\bm{V})} \geq 2^{\lambda(I(U;V)+\epsilon)n} \right) \nonumber \\
 & \stackrel{(\text{d})}{\leq} 2^{-R'n} \left( \frac{\mathbb{E}_{(U,V) \sim Q_{U,V}}\left[ 2^{\lambda i_{Q_{U,V}}(U;V)}\right]}{2^{\lambda(I(U;V)+\epsilon)}} \right)^n \nonumber \\
 & = 2^{-\lambda \left( I(U;V) + \epsilon - \frac{1}{\lambda} \log_2 \mathbb{E}_{(U,V) \sim Q_{U,V}} \left[2^{\lambda i_{Q_{U,V}}(U;V)} \right] \right)n - R'n} \nonumber \\
 & \stackrel{(\text{e})}{=} 2^{-\lambda \left( I(U;V) + \epsilon - D_{\lambda+1}(Q_{U,V}||Q_U Q_V) \right)n - R'n} \nonumber \\
 & = 2^{-(\alpha_{\lambda,\epsilon} + R')n} \label{eq:beta_bd}
 \end{align}
 where (a) follows from the fact that $\bm{U}(w,\mathscr{C}_n)$ is distributed as $Q^n_U$, (b) follows from the definition of $\mathcal{A}_\epsilon$, (c) holds for any $\lambda > 0$, (d) follows from Markov's inequality and the product form of the joint PMF $Q^n_{U,V}$, (e) follows from the definition of R\'{e}nyi divergence of order $\lambda+1$, and where $\alpha_{\lambda,\epsilon} \triangleq \lambda \left( I(U;V) + \epsilon - D_{\lambda+1}(Q_{U,V}||Q_U Q_V) \right)$. 
 
 For $\epsilon>0$, we remark that i) $\alpha_{\lambda,\epsilon}$ tends to $0$ as $\lambda$ tends to $0$, and ii) $\alpha_{\lambda,\epsilon}$ is positive for small enough $\lambda>0$; these follow from the facts that $D_{\lambda+1}(Q_{U,V}||Q_{U}Q_{V})$ is a continuous and non-decreasing function of $\lambda>0$ and that $D_1(Q_{U,V}||Q_{U}Q_{V}) = I(U;V)$. In the sequel, for a given $\epsilon>0$, we let $\lambda>0$ be small enough such that $\alpha_{\lambda,\epsilon} \in (0,R')$. Moving forward, we write $\alpha_{\lambda,\epsilon}$ as simply $\alpha$ when the dependency on $\lambda$ and $\epsilon$ is clear from context.

 \begin{lemma}[{\cite[Theorem 3]{Schmidt1995}}] \label{thm:schmidt}
 Suppose that $\{T_w\}_{w \in \mathcal{W}}$ are random variables that take values in $[0,1]$, and define $T \triangleq \sum_{w \in \mathcal{W}} T_w$ and $\mu \triangleq \mathbb{E}[T]$. For $\tau > 0$, if the variables are $k$-wise independent for some $k \geq k^*(|\mathcal{W}|,\mu,\tau) \triangleq \lceil \frac{\mu \tau}{1 - \frac{\mu}{|W|}} \rceil$, then
 \begin{equation} \nonumber
 \mathbb{P} \left( T \geq \mu(1+\tau) \right) \leq \frac{ {|\mathcal{W}| \choose k^*} \left(\frac{\mu}{|\mathcal{W}|} \right)^{k^*}}{{\mu(1+\tau) \choose k^*}}.
 \end{equation}
 \end{lemma}

 Using the framework of Lemma \ref{thm:schmidt}, we set $T_w$ for each $w \in \mathcal{W}$ to be the $w^{\text{th}}$ term in the sum of (\ref{eq:P2_exp}), i.e.,
 \begin{equation} \nonumber
 \Scale[0.95]{T_w = 2^{-R'n} \mathbb{P}_{\bm{V} \sim Q^n_{V|U}} \left( (\bm{U}(w,\mathscr{C}_n),\bm{V}) \not\in \mathcal{A}_\epsilon | \bm{U} = \bm{U}(w,\mathscr{C}_n) \right),}
 \end{equation}
 and in turn, we have that $T \triangleq \sum_{w \in T_w} T_w = \sum_{\bm{v} \in \mathcal{V}^n} P^{(\mathscr{C}_n)}_{\bm{V},1}(\bm{v})$. Note that the expectation $\mu \triangleq \mathbb{E}_{\mathscr{C}_n} [T]$ is bounded above by $2^{-\alpha n}$ following (\ref{eq:beta_bd}). For a parameter $\beta \in (0,\alpha)$ that will be set later, set $\tau$ such that $\mu(1+\tau) = 2^{(\beta-\alpha)n}$.

 Before applying Lemma \ref{thm:schmidt}, we normalize the random variables $\{T_w\}_{w \in \mathcal{W}}$ to optimize the parameter $k^*$. For some parameter $\theta \in (0,1]$ which we will soon set, define $T'_{w} = \theta 2^{R'n} T_w$ and note that $T'_w \in [0,1]$. Similarly, define the normalized sum $T' = \theta 2^{R'n} T $, its normalized expectation $\mu' = \theta 2^{R'n} \mu$ which is bounded above by  $\theta 2^{(R'-\alpha)n}$, and note that $\mu'(1+\tau)= \theta 2^{(R'+\beta-\alpha)}$. Now consider the quantity $k^*(|\mathcal{W}|,\mu',\tau)$ as a function of $\theta$, and let $n$ be large enough and choose $\theta \in (0,1]$ such that $k^*(|\mathcal{W}|,\mu',\tau)$ is equal to $k$; such a choice exists for fixed $k$ and large enough $n$ since $k^*(|\mathcal{W}|,\mu',\tau) \geq \mu' \tau = \theta 2^{(R'+\beta-\alpha)n} - \mu' \geq \theta 2^{(R'-\alpha)n}(2^{\beta n}-1)$ is tending larger than $k$ for fixed $\theta > 0$ as $n$ tends to infinity following $\alpha < R'$.  

 We apply Lemma \ref{thm:schmidt} to the normalized random variables $\{T'_w\}_{w \in \mathcal{W}}$. We have for large enough $n$
 \begin{align}
 \Scale[0.90]{\mathbb{P}_{\mathscr{C}_n} \left( \sum\limits_{\bm{v} \in \mathcal{V}^n} P^{(\mathscr{C}_n)}_{\bm{V},2}(\bm{v}) \geq 2^{(\beta-\alpha)n} \right)} &=
 \mathbb{P}_{\mathscr{C}_n} \left( T \geq 2^{(\beta - \alpha)n}\right) \nonumber \\
  &\stackrel{(\text{f})}{=} \mathbb{P}_{\mathscr{C}_n}\left(T' \geq \theta 2^{(R+\beta-\alpha)n}\right) \nonumber \\
 &\stackrel{(\text{g})}{\leq} \frac{{2^{R'n} \choose k} \left( \frac{\mu'}{2^{R'n}}\right)^k}{{\theta 2^{(R'+\beta-\alpha)n} \choose k}} \nonumber \\
 & \stackrel{(\text{h})}{\leq} \frac{k^k}{k!} \left( \frac{\mu'}{\theta 2^{(R'+\beta-\alpha)n}} \right)^k  \nonumber \\
 &  \stackrel{(\text{i})}{\leq} \frac{k^k}{k!} 2^{-k \beta n} \label{eq:PV2_1}
 \end{align}
 where (f) follows from the normalization $T' = \theta 2^{R'n} T$, ($g$) follows for large enough $n$ from Lemma \ref{thm:schmidt} and the choice of $\theta$ such that $k^*=k$, (h) follows from the inequalities $\frac{m^k}{k^k}\leq {m \choose k} \leq \frac{m^k}{k!}$ for any $1 \leq k \geq m$, and (i) follows from the bound $\mu' \leq \theta 2^{(R'- \alpha)}$.

 Next, we show that inequality (\ref{eq:suff_2}) holds with high probability over the random code $\mathscr{C}_n$. For $\bm{v} \in \mathcal{V}^n$, expand $\Delta^{(\mathscr{C}_n)}_{\bm{V},1}(\bm{v})$:
 \begin{align}
 &\Delta^{(\mathscr{C}_n)}_{\bm{V},1}(\bm{v}) \triangleq \frac{P^{(\mathscr{C}_n)}_{\bm{V},1}(\bm{v})}{Q^n_{V}(\bm{v})}  \nonumber\\
 &= \sum_{w \in \mathcal{W}} 2^{-R'n} \frac{Q^n_{V|U}(\bm{v} \big| \bm{U}(w,\mathscr{C}_n))}{Q^n_{V}(\bm{v})} \mathds{1} \left\{ \left(\bm{U}(w,\mathscr{C}_n),\bm{v} \right) \in \mathcal{A}_\epsilon \right\}.  \label{eq:D1_exp}
 \end{align}
Note that (\ref{eq:D1_exp}) is a sum of $|\mathcal{W}|=2^{R'n}$ $k$-wise independent terms following that the codewords of $\mathscr{C}_n$ are $k$-wise independent. For $w \in \mathcal{W}$, the expectation of the $w^{\text{th}}$ term in the sum of (\ref{eq:D1_exp}) is 
\begin{align}
& 2^{-R'n} \mathbb{E}_{\mathscr{C}_n} \left[ \frac{Q^n_{V|U}(\bm{v}|\bm{U}(w,\mathscr{C}_n))}{Q^n_V(\bm{v})} \mathds{1} \left\{ \left(\bm{U}(w,\mathscr{C}_n),\bm{v} \right) \in \mathcal{A}_{\epsilon} \right\} \right] \nonumber \\
& \stackrel{(\text{j})}{\leq} 2^{-R'n} \mathbb{E}_{\mathscr{C}_n} \left[ \frac{Q^n_{V|U}(\bm{v}|\bm{U}(w,\mathscr{C}_n))}{Q^n_V(\bm{v})} \right] \nonumber \\
&  \stackrel{(\text{k})}{=} 2^{-R' n} \sum_{\bm{u} \in \mathcal{U}^n} Q^n_{U}(\bm{u}) \frac{Q^n_{V|U}(\bm{v}|\bm{u})}{Q^n_V(\bm{v})} \nonumber \\
& = 2^{-R'n} \label{eq:ED1}
\end{align}
where (j) follows from the trivial bound $\mathds{1}\{\cdot\} \leq 1$ and (k) follows from the distribution of codeword $\bm{U}(w,\mathscr{C}_n) \sim Q^n_U$.
 
 \begin{lemma}[{\cite[Lemma 2.3]{Bellare1994}}] \label{thm:Bellare}
  Let $k \geq 4$ be an even integer. Suppose that $\{T_w\}_{w \in \mathcal{W}}$ are $k$-wise independent random variables that take values in $[0,1]$, and define $T \triangleq \sum_{w \in \mathcal{W}} T_w$ and $\mu \triangleq \mathbb{E}[T]$. For any $\tau > 0$,
  \begin{equation} \nonumber
  \mathbb{P}(T \geq \mu(1+\tau)) \leq 8 \left( \frac{k \mu + k^2}{(\mu \tau)^2} \right)^{k/2}.
  \end{equation}
 \end{lemma}

 Using the framework of Lemma \ref{thm:Bellare}, fix $\bm{v} \in \mathcal{V}^n$ and set $T_w$ for each $w\in \mathcal{W}$ to be
 \begin{equation} \nonumber
 \Scale[0.94]{T_w = 2^{(-I(U;V)-\epsilon)n} \left(\frac{Q^n_{V|U}(\bm{v} | \bm{U}(w,\mathscr{C}_n))}{Q^n_{V}(\bm{v})} \right) \mathds{1} \left\{ \left(\bm{U}(w,\mathscr{C}_n),\bm{v} \right) \in \mathcal{A}_\epsilon \right\}}
 \end{equation}
 which coincides with the $w^{\text{th}}$ term in the sum of (\ref{eq:D1_exp}) normalized by the factor $2^{(R' - I(U;V) - \epsilon)n}$. This normalization factor was chosen to ensure $T_w$ is bounded above by $1$ which follows from that fact that for any $(\bm{u},\bm{v}) \in \mathcal{A}_{\epsilon}$ we have that $\frac{Q^n_{V|U}(\bm{v}|\bm{u})}{Q^n_V(\bm{v})} < 2^{(I(U;V)+\epsilon)n}$. Set $T = \sum_{w \in \mathcal{W}}T_w$ and note that $\mu \triangleq \mathbb{E}_{\mathscr{C}_n}[T]$ is bounded above by $2^{(R' - I(U;V)-\epsilon)n}$ following (\ref{eq:ED1}) and the choice of normalization factor. Finally, set $\tau$ such that $\mu(\tau+1) = 2^{(R' - I(U;V)-\epsilon)n}(1+2^{(\beta-\alpha)n})$ and note that $\mu \tau = 2^{(R'-I(U;V)-\epsilon)n}(1+2^{(\beta-\alpha)n}) - \mu \geq 2^{(R'-I(U;V)-\epsilon+\beta-\alpha)n}$. Applying Lemma \ref{thm:Bellare}, we have that for for even integer $k \geq 4$, small enough $\epsilon>0$ and large enough $n$
 \begin{align}
 \mathbb{P}_{\mathscr{C}_n} \left( \Delta^{(\mathscr{C}_n)}_{\bm{V},1}(\bm{v}) \geq 1 + 2^{(\beta-\alpha)n}\right) &= \mathbb{P}_{\mathscr{C}_n}\left( T \geq \mu(1+\tau)\right) \nonumber \\
 & \hspace{-6em} \stackrel{(\ell)}\leq 8 \left( \frac{k 2^{(R' - I(U;V)-\epsilon)n} + k^2}{2^{2(R'-I(U;V)-\epsilon+\beta-\alpha)n}} \right)^{k/2} \nonumber \\
 & \hspace{-6em} \stackrel{(\text{m})}{\leq} 8 \left( \frac{(k+1) 2^{(R' - I(U;V)-\epsilon)n}}{2^{2(R'-I(U;V)-\epsilon+\beta-\alpha)n}} \right)^{k/2} \nonumber \\
 & \hspace{-6em} = 8 (k+1)^{k/2}  \cdot 2^{-k \eta n}. \label{eq:DV1_1}
 \end{align}
 where ($\ell$) follows from Lemma \ref{thm:Bellare} and the bounds $\mu \leq 2^{(R'-I(U;V)-\epsilon)n}$ and $\mu \tau \geq 2^{(R'-I(U;V)-\epsilon+\beta-\alpha)n}$, and (m) follows for small enough $\epsilon>0$ and large enough $n$ such that $k 2^{(R'-I(U;V)-\epsilon)n} >> k^2$, and where
 
 \begin{align}
 &\eta = \frac{R'-I(U;V)-\epsilon +2(\beta-\alpha)}{2} \label{eq:eta}
 \end{align}
 In turn, by a simple union bound over all $\bm{v} \in \mathcal{V}^n$, and by letting $k \geq 4$ be an even integer, $\epsilon>0$ be small enough and $n$ be large enough,
 \begin{equation} \label{eq:DV1_2}
 \begin{aligned} 
 &\mathbb{P}_{\mathscr{C}_n}\left( \exists \bm{v} \in \mathcal{V}^n \text{ s.t. } \Delta^{(\mathscr{C}_n)}_{\bm{V},1}(\bm{v}) \geq 1 + 2^{(\beta-\alpha)n} \right) \\
 & \hspace{2em} \leq 8k (k+1)^{k/2} \cdot  2^{-(k \eta_1 +\log_2|\mathcal{V}|)n}.
 \end{aligned}
 \end{equation}

 To complete the proof, we put together the above results and apply the sufficient condition (Lemma \ref{thm:sscl_LI}). In the framework of Lemma \ref{thm:sscl_LI}, we set $\pi_1 = \alpha - \beta$. If $\pi_1>0$, then it follows from Lemma \ref{thm:sscl_LI} that the inequalities (\ref{eq:suff_1}) and (\ref{eq:suff_2}) hold with probability at least $1-\pi_0$ where
 $$\pi_0 = \frac{k^k}{k!}2^{-k \beta n} + 8k (k+1)^{k/2} \cdot 2^{(-k\eta+\log_2 |\mathcal{V}|)n}$$ where the expression for $\pi_0$ follows from (\ref{eq:PV2_1}) and (\ref{eq:DV1_2}) together with a simple union bound. 
 
 The last step is to show that for some choice of the free parameters $\epsilon>0$, $\lambda>0$ and $\beta \in (0,\alpha)$ we have that $\pi_1 > 0$ and $\pi_0 = 2^{-k \Omega(n) + n \log_2 |\mathcal{V}|}$. Recall that for a fixed $\epsilon>0$, $\alpha = \alpha_{\lambda,\epsilon}$ tends to $0$ as $\lambda$ tends to $0$, and $\alpha_{\lambda,\epsilon}$ is positive for small enough $\lambda>0$. Furthermore, recall that $R' > I(U;V)$ by assumption, and thus, $\eta$ given by (\ref{eq:eta}) is positive for small enough $\epsilon>0$, small enough $\alpha_{\lambda,\epsilon}>0$, and any $\beta \in (0,\alpha_{\lambda,\epsilon})$.  Thus, given even $k\geq 4$, we can pick $\epsilon>0$ small enough, and in turn, pick $\lambda>0$ small enough such that both $\alpha_{\lambda,\epsilon}$ and $\eta_1$ are positive. In turn, picking $\beta \in (0,\alpha_{\lambda,\epsilon})$ ensures that $\alpha_{\lambda,\epsilon} - \beta >0$ and thus $\pi_1>0$. Thus, $\pi_0 = 2^{-k\Omega(n) + \log_2 |\mathcal{V}|}$. This completes the proof of Lemma \ref{thm:suff_sscl_LI}. 
 
 

 \section{Proof of Theorem \ref{thm:secret}} \label{sec:secret_proof}

\textit{Setup:} Let $p \in [0,1/2]$ and $r \in [0,1]$ such that $1 - H_2(p)- r > 0$. For $\epsilon>0$ and $\epsilon' \in (0,\epsilon)$, let $R = 1 - H_2(p) - r - \epsilon$ and $R' = r + \epsilon'$. Let $k$ be a positive integer to be set in the proof. The goal of the proof is to show that for large enough $k$ constant in $n$ and for large enough $n$, there exists an $[n,Rn,R'n,k]$ pseudolinear code $\mathcal{C}_n$ such that both $\mathrm{Sem}(\mathcal{C}_n) = 2^{-\Omega(n)}$ and $P^{\mathrm{max}}_{\mathrm{error}}(\mathcal{C}_n) = o(1)$. 

\textit{Encoding:} Alice uses an $[n,Rn,R'n]$ code $\mathcal{C}_n = \{ \bm{x}(m,w) \}_{(m,w) \in \mathcal{M} \times \mathcal{W}}$ to encode her message $M$. That is, for a message distribution $P_M \in \mathcal{P}(\mathcal{M})$, Alice draws $M \sim P_M$ and $W \sim \mathrm{Unif}(\mathcal{W})$ and transmits $\bm{x}(M,W)$. 

\textit{Decoding:} Upon receiving the channel output $\bm{y}$, Bob performs min-distance decoding by choosing the message estimate $\widehat{m}$ and key estimate $\widehat{w}$ such that $$(\widehat{m},\widehat{w}) = \arg \min_{(m,w) \in \mathcal{M} \times \mathcal{W}} d_H\left(\bm{x}(m,w),\bm{y} \right)$$
where $d_H$ denotes the Hamming distance.



\subsection{Code Distribution}

We show the existence of a good code via a random coding argument. As our random code distribution, we will use the following distribution over $[n,Rn,R'n,k]$ pseudolinear codes.

\begin{definition}[Random Code Dist.]
Let $F[n,Rn,R'n,k]$ be the distribution over all $[n,Rn,R'n,k]$ pseudolinear codes where the parity check matrix $H$ (c.f. Definition \ref{def:plc}) is fixed and the generator matrix $G$ is chosen uniformly from $\{0,1\}^{\ell \times n}$.
\end{definition}

The following property of $F[n,Rn,R'n,k]$ is useful. 

\begin{lemma}[{\cite[Lemma 9.1]{Guruswami2001}}]
The codewords of $\mathscr{C}_n \sim F[n,Rn,R'n,k]$ are uniformly distributed over $\{0,1\}^n$ and are $k$-wise independent.
\end{lemma}

\subsection{Secrecy Analysis}

For a given $\mathcal{S} \in \mathscr{S}$, let $Q^{(\mathcal{S})}_{\bm{Z}}$ denote the PMF of the adversary's observation $\bm{Z} \in \{0,1\}^{rn}$ when Alice sends a \textit{random $n$-bit sequence} $\bm{X} \sim Q^n_X \triangleq \mathrm{Unif}(\{0,1\}^n)$ through the channel. We have that
\begin{equation} \label{eq:QZ}
Q^{(\mathcal{S})}_{\bm{Z}}(\bm{z}) = Q_{\bm{X}(\mathcal{S})}(\bm{z}) = Q^{rn}_X(\bm{z}), \text{ for all } \bm{z} \in \{0,1\}^{rn}.
\end{equation} 
Furthermore, for an $[n,Rn,R'n]$ code $\mathcal{C}_n$, let $P^{(\mathcal{C}_n,\mathcal{S})}_{M,\bm{Z}}$ denote the joint PMF of message $M$ and observation $\bm{Z}$ when Alice sends the codeword $\bm{x}(M,W,\mathcal{C}_n)$ through the channel. Then for marginal PMF $P_M \in \mathcal{P}(\mathcal{M})$,
\begin{align}
I_{\mathcal{C}_n}\left( M; \bm{Z} \right) &\triangleq D \left( P^{(\mathcal{C}_n,\mathcal{S})}_{M,\bm{Z}} || P_M P^{(\mathcal{C}_n,\mathcal{S})}_{\bm{Z}} \right) \nonumber \\
& \stackrel{(\text{a})}{=}D\left(P^{(\mathcal{C}_n,\mathcal{S})}_{M,\bm{Z}} || P_M Q^{(\mathcal{S})}_{\bm{Z}} \right) - D\left(P^{(\mathcal{C}_n,\mathcal{S})}_{\bm{Z}} || Q^{(\mathcal{S})}_{\bm{Z}} \right) \nonumber \\
& \stackrel{(\text{b})}{\leq} D\left(P^{(\mathcal{C}_n,\mathcal{S})}_{M,\bm{Z}} || P_M Q^{(\mathcal{S})}_{\bm{Z}} \right) \nonumber \\
& \leq \sum_{m \in \mathcal{M}} P_M(m) \max_{m' \in \mathcal{M}} D\left(P^{(\mathcal{C}_n,\mathcal{S})}_{\bm{Z}|M=m'} || Q^{(\mathcal{S})}_{\bm{Z}} \right) \nonumber \\
& = \max_{m \in \mathcal{M}} D\left(P^{(\mathcal{C}_n,\mathcal{S})}_{\bm{Z}|M=m} || Q^{(\mathcal{S})}_{\bm{Z}} \right) \label{eq:sa_1}
\end{align}
where (a) follows from the relative entropy chain rule and (b) follows from the property $D(\cdot || \cdot) \geq 0$. Thus,
\begin{align}
\mathrm{Sem}(\mathcal{C}_n) &= \max_{P_M \in \mathcal{P}(\mathcal{M}), \mathcal{S} \in \mathscr{S}} I_{\mathcal{C}_n}(M;\bm{Z}) \nonumber \\
& \stackrel{(\text{c})}{\leq} \max_{\mathcal{S} \in \mathscr{S}} \max_{m \in \mathcal{M}} D\left(P^{(\mathcal{C}_n,\mathcal{S})}_{\bm{Z}|M=m} || Q^{(\mathcal{S})}_{\bm{Z}} \right)  \nonumber \\
& \stackrel{(\text{d})}{=} \max_{\mathcal{S} \in \mathscr{S}} \max_{m \in \mathcal{M}} D\left(P^{(\mathcal{C}_n,\mathcal{S})}_{\bm{Z}|M=m} || Q^{rn}_{X} \right) \label{eq:sa_2}
\end{align}
where (c) follows from (\ref{eq:sa_1}) and (d) follows from (\ref{eq:QZ}).

  \begin{figure}[t]
  \includegraphics[width=\columnwidth]{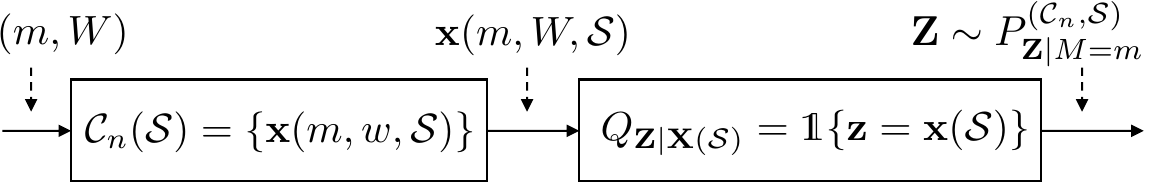}
  \caption{A mapping of the quantities in (\ref{eq:sa_2}) to the soft-covering problem of Fig. \ref{fig:coding_prob_0}.}
  \label{fig:coding_problem}
  \end{figure}

Consider the relative entropy $D \left( P^{(\mathcal{C}_n,\mathcal{S})}_{\bm{Z}|M=m} || Q^{rn}_X \right)$ in the framework of the soft-covering lemma for $k$-wise independent codewords (Lemma \ref{thm:sscl_LI}), as illustrated in Fig. \ref{fig:coding_problem}. Here, $(m,W)$ is uniformly drawn from a message-key product set $\{m\} \times \mathcal{W}$ of rate $R'/r$, i.e., $|\{m\} \times \mathcal{W}| = 2^{R'n} = 2^{rn \frac{R'}{r}}$. Since rate $R' / r = (r+\epsilon')/r$ is greater than $I(X;Z) = 1$, it follows from Lemma \ref{thm:sscl_LI} that there exists $\gamma_0>0$ and $\gamma_1>0$ such that for even integer $k\geq 4$ and large enough $n$,
\begin{equation} \label{eq:sa_3}
\mathbb{P}_{\mathscr{C}_n} \left( D \left( P^{(\mathscr{C}_n,\mathcal{S})}_{\bm{Z}|M=m} || Q^{rn}_X \right) > 2^{-\gamma_1 rn} \right) \leq 2^{(-k \gamma_0 +1)rn}.
\end{equation}
In turn,
\begin{align}
&\mathbb{P}_{\mathscr{C}_n} \left( \mathrm{Sem}(\mathscr{C}_n) > 2^{-\gamma_1 rn} \right) \nonumber \\
&\stackrel{(\text{e})}{\leq} \mathbb{P}_{\mathscr{C}_n} \left( \max_{\mathcal{S} \in \mathscr{S}} \max_{m \in \mathcal{M}} D\left(P^{(\mathscr{C}_n,\mathcal{S})}_{\bm{Z}|M=m} || Q^{rn}_{X} \right) > 2^{-\gamma_1 r n} \right) \nonumber \\
& \leq \mathbb{P}_{\mathscr{C}_n} \left( \bigcup_{\mathcal{S} \in \mathscr{S}} \bigcup_{m \in \mathcal{M}}  \left\{ D\left(P^{(\mathscr{C}_n,\mathcal{S})}_{\bm{Z}|M=m} || Q^{rn}_{X} \right) > 2^{-\gamma_1 r n} \right\}  \right) \nonumber \\
& \stackrel{(\text{f})}{\leq} 2^{(-k \gamma_0 r + r + R+1)n} \label{eq:secure}
\end{align}
where (e) follows from (\ref{eq:sa_2}), and (f) follows for large enough $n$ from a simple union bound, the inequality $|\mathscr{S}| = {n \choose rn} \leq 2^{n}$ and (\ref{eq:sa_3}).

\subsection{Reliability Analysis}

 Unlike the above secrecy analysis, the reliability analysis requires additional structure of the code $\mathscr{C}_n$ beyond the $k$-wise independence property. In particular, we will use the pseudolinear structure of $\mathscr{C}_n$. We restate a reliability result of \cite{Ruzomberka2023} without proof. For a code $\mathcal{C}_n$ and a message $m \in \mathcal{M}$, define the probability of decoding error conditioned on $M=m$ as $$P^{(m)}_{\mathrm{error}}(\mathcal{C}_n) \triangleq \mathbb{P}(\widehat{M} \neq m | M=m)$$
 where the probability is w.r.t. $W \sim \mathrm{Unif}(\mathcal{W})$ and the adversary's choice of bit read/flip locations.

 \begin{lemma}[{\cite[Theorem 1]{Ruzomberka2023}}] \label{thm:Ruz2023}
 Suppose that $p \in (0,1/2)$ and $r< 1 - H_2(p)$. If the key rate $R' > r$ and the sum rate $R+R' < 1 - H_2(p)$, then for large enough (but fixed) $k$ and any fixed $\delta>0$, there exists $\gamma_2>0$ such that for large enough $n$ and any $m \in \mathcal{M}$,
 \begin{equation} 
 \mathbb{P}_{\mathscr{C}_n} \left( P^{(m)}_{\mathrm{error}}(\mathscr{C}_n) > \delta \right) \leq 2^{-k\gamma_2 n}.
 \end{equation}
 \end{lemma}

 We apply Lemma \ref{thm:Ruz2023} to bound the maximum probability of error $P^{\mathrm{max}}_{\mathrm{error}}(\mathscr{C}_n) \triangleq \max_{m \in \mathcal{M}} P^{(m)}_{\mathrm{error}} (\mathscr{C}_n)$. Note that our choice of $\epsilon$ and $\epsilon'$ ensures that $R'>r$ and $R+R' < 1-H_2(p)$. Also, we have that $R < 1 - H_2(p) - r$. Thus, for $\delta > 0$,
 \begin{align}
 \mathbb{P}_{\mathscr{C}_n} \left( P^{\mathrm{max}}_{\mathrm{error}(\mathscr{C}_n)} > \delta \right) &\triangleq \mathbb{P}_{\mathscr{C}_n} \left( \max_{m \in \mathcal{M}} P^{(m)}_{\mathrm{error}}(\mathscr{C}_n) > \delta \right) \nonumber \\
 & \stackrel{(\text{g})}{\leq} \sum_{m \in \mathcal{M}} \mathbb{P}_{\mathscr{C}_n} \left( P^{(m)}_{\mathrm{error}}(\mathscr{C}_n) > \delta \right) \nonumber \\
 & \stackrel{(\text{h})}{\leq} 2^{(-k \gamma_2  + 1 - H_2(p) - r)n} \label{eq:reliable}
 \end{align}
 where (g) follows from a union bound and (h) follows for large enough $k$ and for large enough $n$ via Lemma \ref{thm:Ruz2023}.

 \subsection{Combining Secrecy and Reliability Analysis}

 To complete the proof, we combine the secrecy and reliability analysis. For large enough $k$ and $k$ even, and for large enough $n$,
 \begin{align}
 &\mathbb{P}_{\mathscr{C}_n}( \{\mathrm{Sem}(\mathscr{C}_n) > 2^{-\gamma r n}\} \cup \{ P^{\mathrm{max}}_{\mathrm{error}}(\mathscr{C}_n) > \delta \}) \nonumber \\
 &\leq 2^{(-k \gamma_0 r + 2r + R)n} + 2^{(-k \gamma_2  + 1 - H_2(p) - r)n} \label{eq:last_eq}
 \end{align}
 following both (\ref{eq:secure}), (\ref{eq:reliable}) and a simple union bound. In summary, for large enough $k$ and $k$ even (which is constant in $n$) and large enough $n$, we have that (\ref{eq:last_eq}) is less than $1$, and in turn, there exists an $[n,Rn,R'n,k]$ pseudolinear code $\mathcal{C}_n$ such that $\mathrm{Sem}(\mathcal{C}_n) \leq 2^{-\gamma_1 r n}$ and $P^{\mathrm{max}}_{\mathrm{error}}(\mathcal{C}_n) \leq \delta$.

 \section{Conclusion}

 We showed that random pseudolinear codes achieve the best known lower bound of the semantic secrecy capacity of the binary adversarial wiretap channel of type II. A necessary condition on the non-linearity of a capacity achieving code was also shown. One possible avenue for future research is to apply further derandomization techniques to our random codes, e.g., in the spirit of \cite{Guruswami2022}. The goal here is to replace random pseudolinear codes with a significantly derandomized class that can maintain the same error-correction and secrecy power while being more amendable to efficient decoding algorithms.
 
 \ifthenelse{\boolean{extend_v}}{ 
 \appendices

 \section{Linear Coset Coding Schemes} \label{sec:LCCS}

 In this appendix, we prove that the linear coset coding scheme of Ozarow and Wyner \cite{Ozarow1986} is not semantically-secret for any positive message rate. We first define coset coding.

The linear coset coding scheme, proposed in \cite{Ozarow1986}, is as follows: Let $R>0$ be the message rate. For blocklength $n$, let $H$ be the $Rn \times n$ parity check matrix of some $[n,n-Rn]$ binary linear code. \textit{Encoding:} Suppose that Alice wants to transmit a message $m \in \{0,1\}^{Rn}$. Alice encodes $m$ by choosing the $n$ bit codeword $\bm{x}$ randomly and uniformly from the set of solutions $\{\bm{x}' \in \{0,1\}^n: \bm{x}' H^T = \bm{m}$\} and transmits $\bm{x}$ over the (noiseless) $(0,r)$-AWTC II. \textit{Decoding:} Upon receiving $\bm{x}$, Bob performs decoding by choosing the message estimate $\widehat{m} = \bm{x}H^T$. It is easy to show that the above linear coset coding scheme is an $[n,Rn,(1-R)n]$ linear code.

We prove the following result.
\begin{lemma} \label{thm:LCCS}
Let rate $R>0$. For large enough $n$, any $[n,Rn,(1-R)n]$ binary code $\mathcal{C}_n$ that is a linear coset coding scheme has semantic leakage $\mathrm{Sem}(\mathcal{C}_n) \geq 1$.
\end{lemma}

For any $R>0$, let $\mathcal{C}_n$ be an $[n,Rn,(1-R)n]$ binary code that is a linear coset coding scheme and let $H$ be the corresponding $Rn \times n$ parity check matrix. Suppose that Alice's message is uniformly distributed over $\{0,1\}^n$. To prove Lemma \ref{thm:LCCS}, we will use the following result due to Ozarow and Wyner.
\begin{lemma}[{\cite[Lemma 4]{Ozarow1986}}] \label{thm:Ozarow4} For an index set $\mathcal{I} \subseteq [n]$, let $H(\mathcal{I})$ denote the $|\mathcal{I}|$ columns of $H$ indexed by $\mathcal{I}$. The adversary's equivocation is
\begin{equation}
\Delta \triangleq \min_{\mathcal{S} \in \mathscr{S}} H(M|\bm{Z}) = \min_{\mathcal{I} \subseteq [n]: |\mathcal{I}| = (1-r)n} \mathrm{rank} \left( H(\mathcal{I}) \right).
\end{equation}
\end{lemma}

Recall the following definitional inequalities: 
\begin{align}
\mathrm{Sem}(\mathcal{C}_n) &\geq \max_{\mathcal{S} \in \mathscr{S}} I_{\mathcal{S}}(M;\bm{Z}) \nonumber \\
& = H(M) - \min_{\mathcal{S} \in \mathscr{S}} H(M|\bm{Z}) = Rn - \Delta. \nonumber
\end{align}
Thus, to show that $\mathrm{Sem}(\mathcal{C}_n) \geq 1$ for large enough $n$, it is sufficient to show that $\Delta \leq Rn -1$.

Let $n$ be large enough and suppose by contradiction that $\Delta = Rn$. By Lemma \ref{thm:Ozarow4}, we have that $\mathrm{rank}(H(\mathcal{I})) = Rn$ for every set $\mathcal{I} \subseteq [n]$ s.t. $|\mathcal{I}| = (1-r)n$. This in turn by the definition of $H$ implies that the $[n,(1-R)n]$ binary code with parity check matrix $H$ has minimum distance, denoted $d_{\mathrm{min}}$, of at least $Rn+1$. However, by the Plotkin bound of Lemma \ref{thm:Plotkin}, we have that $1-R \leq 1 - 2 \frac{d_{\mathrm{min}}}{n} + o(1)$, or equivalently, $d_{\mathrm{min}} \leq \frac{Rn}{2} + o(n)$. Thus, for $n$ large enough such that the $o(n)$ term is negligible, we have a contradiction. This completes the proof of Lemma \ref{thm:LCCS}.

 \section{Discussion of Assumption \ref{ass:rank}} \label{sec:assrank_proof}

 We show that if the generator matrix $G$ of an $[n,Rn]$ linear code $\mathcal{C}_n$ is not full rank, then either the probability of decoding error is large such that $P^{\mathrm{max}}_{\mathrm{error}}(\mathcal{C}_n) \geq 1/2$ or both $\mathcal{W}$ and $G$ can be replaced with a smaller key set $\mathcal{W}'$ and generator matrix $G'$, respectively, without changing the code. Let $\mathcal{C}_n$ be an $[n,Rn]$ linear code and suppose that $G$ is not full rank.
 
 Suppose that $G_W$ is full rank. Since the channel is noiseless, Bob's received sequence is guaranteed to be a codeword in $\mathcal{C}_n$. Suppose that Bob receives the codeword $\bm{c} \in \mathcal{C}_n$. From Bob's perspective, the set of all possible message-key pairs that Alice could have sent is
 \begin{align}
 \mathcal{M}_{\bm{c}} &= \{ (m,w) \in \mathcal{M} \times \mathcal{W}: \begin{bmatrix} m & w \end{bmatrix} G = \bm{c}\} \nonumber \\
  &= \{ (m,w) \in \mathcal{M} \times \mathcal{W}: m G_M + w G_W = \bm{c}\}. \nonumber
 \end{align}
 Since the mapping $G:\{0,1\}^{(R+R')n} \rightarrow \{0,1\}^{n}$ is a linear transformation, the number of pairs in $\mathcal{M}_{\bm{c}}$ is $|\mathcal{M}_{\bm{c}}|= 2^{\mathrm{nullity}(G)} = 2^{(R+R')n - \mathrm{rank}(G)}$ where the second equality follows from the rank-nullity theorem. In turn, since $\mathrm{rank}(G) < (R+R')n$, it follows that $|\mathcal{M}_{\bm{c}}| \geq 2$. Now consider two unique pairs in $\mathcal{M}_{\bm{c}}$, say $(m_1,w_1)$ and $(m_2,w_2)$. We show that $m_1 \neq m_2$ by considering 2 cases. (Case 1): Suppose that $w_1 = w_2$. Then $m_1 \neq m_2$ by the uniqueness of the pairs. Done. (Case 2): Suppose instead that $w_1 \neq w_2$. Since $G_W$ is full rank, we have that $(w_1+w_2)G_W \neq 0$. In turn, $[m_1 w_1]G = [m_2 w_2]G$ implies that $(m_1+m_2) G_M = (w_1+w_2) G_W \neq 0$, and thus, $m_1 \neq m_2$. Done. In summary, upon receiving $\bm{c}$, Bob finds that at least 2 messages could be Alice's message. Thus, for PMFs $P_M = \mathrm{Unif}(\mathcal{M})$ and $P_W = \mathrm{Unif}(\mathcal{W})$,
 \begin{align}
 &P^{\mathrm{max}}_{\mathrm{error}}(\mathcal{C}_n) \geq \mathbb{P}_{(M,W) \sim P_M P_W} \left( \widehat{M} \neq M \right) \nonumber \\
 & = \sum_{\bm{c} \in \mathcal{C}_n} \mathbb{P}_{(M,W) \sim P_M P_W} \left( \widehat{M} \neq M \Big| \text{Bob RXs }\bm{c} \right) \frac{1}{|\mathcal{C}_n|} \nonumber \\
 & \geq 1/2. \nonumber
 \end{align}

 Suppose instead that $G_W$ is not full rank. Then each $(R'n)$-bit sequence in the rowspace of $G_W$ corresponds to multiple (i.e., redundant) keys in $\mathcal{W}$. Hence, we can eliminate this redundancy by shortening the key $w$ from $R'n$ bits to $\mathrm{rank}(G_W)$ bits and replacing $G_W$ with full rank matrix $G'_W$ that has $\mathrm{rowspace}(G'_W) = \mathrm{rowspace}(G_W)$ without changing the code $\mathcal{C}_n$.
 }{}

\bibliographystyle{IEEEtran}
\bibliography{refs}

\begin{thebibliography}{10}
\providecommand{\url}[1]{#1}
\csname url@samestyle\endcsname
\providecommand{\newblock}{\relax}
\providecommand{\bibinfo}[2]{#2}
\providecommand{\BIBentrySTDinterwordspacing}{\spaceskip=0pt\relax}
\providecommand{\BIBentryALTinterwordstretchfactor}{4}
\providecommand{\BIBentryALTinterwordspacing}{\spaceskip=\fontdimen2\font plus
\BIBentryALTinterwordstretchfactor\fontdimen3\font minus
  \fontdimen4\font\relax}
\providecommand{\BIBforeignlanguage}[2]{{%
\expandafter\ifx\csname l@#1\endcsname\relax
\typeout{** WARNING: IEEEtran.bst: No hyphenation pattern has been}%
\typeout{** loaded for the language `#1'. Using the pattern for}%
\typeout{** the default language instead.}%
\else
\language=\csname l@#1\endcsname
\fi
#2}}
\providecommand{\BIBdecl}{\relax}
\BIBdecl

\bibitem{Ozarow1986}
L.~H. Ozarow and A.~D. Wyner, ``{Wire-tap channel II},'' \emph{AT\&T Bell
  Laboratories Technical Journal}, vol.~63, pp. 2135--2157, 1986.

\bibitem{Wang2016_2}
P.~Wang and S.~Safavi-Naini, ``{A model for adversarial wiretap channels},''
  \emph{IEEE Trans. Inf. Theory}, vol.~62, pp. 970 -- 983, 2016.

\bibitem{Goldwasser1984}
S.~Goldwasser and S.~Micali, ``{Probabilistic encryption},'' \emph{J. Computer
  and System Science}, vol.~28, pp. 270--299, 1984.

\bibitem{Bellare1997}
M.~Bellare, A.~Desai, E.~Jokipii, and P.~Rogaway, ``{A concrete security
  treatment of symmetric encryption},'' in \emph{Proc. IEEE Symp. Foundations
  of Computer Science}, Aug 1997, pp. 394--403.

\bibitem{Bellare2012}
M.~Bellare, S.~Tessaro, and A.~A. Vardy, ``{A cryptographic treatment of the
  wiretap channel},'' in \emph{Proc. Adv. Cryptol. (CRYPTO)}, Aug 2012, pp.
  1--31.

\bibitem{Wang2016}
C.~Wang, ``{On the capacity of the binary adversarial wiretap channel},'' in
  \emph{Proc. of the 54th Annual Allerton Conference on Communications, Control
  and Computing}, 2016.

\bibitem{Ozel2013}
O.~Ozel and U.~Sennur, ``{Wiretap channels: implications of the more capable
  condition and cyclic shift symmetry},'' \emph{IEEE Trans. Inf. Theory},
  vol.~59, pp. 2153 -- 2164, 2013.

\bibitem{Guruswami2001_2}
V.~Guruswami and P.~Indyk, ``{Expander-based constructions of efficiently
  decodable codes},'' in \emph{Proc. IEEE Symp. on Foundations of Computer
  Science}, 2001, p. 658–667.

\bibitem{Wyner1975}
A.~D. Wyner, ``{The common information of two dependent random variables},''
  \emph{IEEE Trans. Inf. Theory}, vol.~21, pp. 163--179, 1975.

\bibitem{Goldfeld2016}
Z.~Goldfeld, P.~Cuff, and H.~H. Permuter, ``{Semantic-security capacity for
  wiretap channels of type II},'' \emph{IEEE Trans. Inf. Theory}, vol.~62,
  no.~7, pp. 3863--3879, 2016.

\bibitem{Mahdavifar2011}
H.~Mahdavifar and A.~Vardy, ``{Achieving the secrecy capacity of wiretap
  channels using polar codes},'' \emph{IEEE Trans. Inf. Theory}, vol.~57, pp.
  6428--6443, 2011.

\bibitem{Bassalygo1965}
L.~A. Bassalygo, ``{New upper bounds for error-correcting codes},''
  \emph{Problems of Information Transmission}, vol.~1, pp. 32--35, 1965.

\bibitem{Shannon1967_I}
C.~E. Shannon, R.~G. Gallager, and E.~R. Berlekamp, ``{Lower bounds to error
  probability for coding on discrete memoryless channels. I},''
  \emph{Information and Control}, vol.~10, pp. 65--103, 1967.

\bibitem{Shannon1967_II}
------, ``{Lower bounds to error probability for coding on discrete memoryless
  channels. II},'' \emph{Information and Control}, vol.~10, pp. 522--552, 1967.

\bibitem{Guruswami2001}
\BIBentryALTinterwordspacing
V.~Guruswami, ``List decoding of error-correcting codes,'' Ph.D. dissertation,
  Massachusetts Institute of Technology, 2001. [Online]. Available:
  \url{http://hdl.handle.net/1721.1/8700}
\BIBentrySTDinterwordspacing

\bibitem{Macwilliams1977}
F.~J. MacWilliams and N.~J.~A. Sloane, \emph{{The Theory of Error-Correcting
  Codes}}.\hskip 1em plus 0.5em minus 0.4em\relax North-Holland Publishing
  Company, 1977.

\bibitem{Cheraghchi2012}
M.~Cheraghchi, F.~Didier, and A.~Shokrollahi, ``{Invertible extractors and
  wiretap protocols},'' \emph{IEEE Trans. Inf. Theory}, vol.~58, pp.
  1254--1274, 2012.

\bibitem{Chou2022}
R.~A. Chou, ``{Explicit wiretap channel codes via source coding, universal
  hashing, and distribution approximation, when the channels’ statistics are
  uncertain},'' \emph{IEEE Trans. Inf. Forensics Security}, vol.~18, pp.
  117--132, 2022.

\bibitem{Guruswami2016}
V.~Guruswami and A.~Smith, ``{Optimal rate code constructions for
  computationally simple channels},'' \emph{Journal of the ACM}, vol.~63,
  no.~4, pp. 1--37, 2016.

\bibitem{Sharifian2018}
S.~Sharifian, F.~Lin, and R.~Safavi-Naini, ``{Hash-then-encode: a modular
  semantically secure wiretap code},'' in \emph{Proc. Workshop on Comm.
  Security}, July 2018, pp. 49--63.

\bibitem{Ruzomberka2023}
E.~Ruzomberka, H.~Nikbakht, C.~G. Brinton, and H.~V. Poor, ``{On pseudolinear
  codes for correcting adversarial errors},'' in \emph{Proc. IEEE Symp.
  Foundations of Computer Science}, to appear, available on arXiv.

\bibitem{Hoffman1971}
K.~Hoffman and R.~Kunze, \emph{{Linear Algebra}}, 2nd~ed.\hskip 1em plus 0.5em
  minus 0.4em\relax Prentice-Hall, 1971.

\bibitem{Plotkin1960}
M.~Plotkin, ``{Binary codes with specified minimum distance},'' \emph{IRE
  Trans. Inf. Theory}, vol.~6, pp. 445--450, 1990.

\bibitem{Rudra2007}
A.~Rudra, ``Lecture notes in error correcting codes: Combinatorics, algorithms
  and applications,''
  \url{https://cse.buffalo.edu/faculty/atri/courses/coding-theory/lectures/lect16.pdf},
  October 2007.

\bibitem{Schmidt1995}
J.~P. Schmidt, A.~Siegel, and A.~Srinivasan, ``{Chernoff-Hoeffding bounds for
  applications with limited independence},'' \emph{SIAM Journal on Discrete
  Mathematics}, vol.~8, no.~2, pp. 223--250, 1995.

\bibitem{Bellare1994}
M.~Bellare and J.~Rompel, ``{Randomness-efficient oblivious sampling},'' in
  \emph{Proc. IEEE Symp. Symposium on Foundations of Computer Science}, Aug
  1994, pp. 276--287.

\bibitem{Guruswami2022}
V.~Guruswami and J.~Mosheiff, ``{Punctured low-bias codes behave like random
  linear codes},'' in \emph{Proc. IEEE Symp. Foundations of Computer Science},
  2022, pp. 36--45.

\end{thebibliography}

%
\IEEEpeerreviewmaketitle

\end{document}